%% file: paper.tex
\documentclass[a4paper,USenglish]{lipics}

\usepackage{amsthm,amsmath, amssymb}
\usepackage{todonotes}
\usepackage{mathrsfs}

\input{pre}

\title{Data Structure Lower Bounds for Document Indexing Problems}
\author{Peyman Afshani}
\author{Jesper Sindahl Nielsen}
\affil{MADALGO\footnote{ Research funded by MADALGO, Center for Massive Data Algorithmics, a
    Center of the Danish National Research Foundation, grant DNRF84}\\
  Department of Computer Science, Aarhus University, Denmark}
\authorrunning{P. Afshani and J. S. Nielsen}
\Copyright{Peyman Afshani and Jesper Sindahl Nielsen}

\begin{document}

\maketitle

\input{abstract}
\thispagestyle{empty}
%\pagebreak

%\setcounter{page}{1}

\section{Introduction}
\input{intro}

\input{gapped_pattern}

\input{query_lb}

\input{two_pattern}

\input{conclusions}

\paragraph*{Acknowledgements.}
We thank Karl Bringmann for simplifying the construction of our documents for the 2P problem. 
We thank Moshe Lewenstein for bringing the WCI problem to our attention. 
We also thank the anonymous referees for pointing us to~\cite{DietzMRU95}.

\bibliography{bib}{}
\bibliographystyle{abbrv}

\newpage

\input{appendix_main}

\end{document}

%% file: pre.tex
\usepackage[T1]{fontenc}
\usepackage[utf8]{inputenc}
\usepackage{thm-restate}
\usepackage{relsize}

\renewcommand{\O}{\mathcal{O}}
\renewcommand{\P}{\mathcal{P}}

\newtheorem*{theorem*}{Theorem}
\newtheorem{fact}{Fact}
\newtheorem*{fact*}{Fact}
\newtheorem{ob}{Observation}
\newcommand{\rlem}[1]{Lemma~\ref{lemma:#1}}

\newcommand{\w}{\mathbf{w}}
\newcommand{\U}{\mathcal{U}}

\newcommand{\Q}{\mathcal{Q}}

\newcommand{\para}[1]{\vspace{2mm} \noindent \textbf{#1}}

\newcommand{\ignore}[1]{}

\newcommand{\0}{\mathbf{0}}
\newcommand{\1}{\mathbf{1}}
\newcommand{\2}{\mathbf{2}}
\newcommand{\bi}{\mathbf{i}}
\newcommand{\bs}{\textbf{\#}}

\newcommand{\true}{1}

\newcommand{\toappendix}[4]{
	\ifnum#2=1
	\newcommand{#1}{#3 #4}
	\else
	\newcommand{#1}{}
	#4
	\fi
}

%% file: abstract.tex
\begin{abstract}
	We study data structure problems related to document indexing and pattern matching queries
	and our main contribution is to show that the pointer machine model of computation 
	can be extremely useful in proving high and unconditional lower bounds that cannot be obtained in 
	any other known model of computation with the current techniques. 
	Often our lower bounds match the known space-query time trade-off curve and in fact for all 
	the problems considered, there is a very good and reasonable match between the
	our lower bounds and the known upper bounds, at least for some choice of input parameters.
 
	The problems that we consider are \emph{set intersection} queries (both the \emph{reporting} variant and the 
	\emph{semi-group counting} variant), indexing a set of documents for \emph{two-pattern} queries, or
	\emph{forbidden-pattern} queries, or \emph{queries with wild-cards}, and indexing
	an input set of \emph{gapped-patterns} (or \emph{two-patterns})	
	to find those matching a document given at the query time.

 \end{abstract}

%% file: intro.tex
In this paper we study a number of data structure problems related to document indexing and prove
space and query time lower bounds that in most cases are almost tight. 
Unlike many of the previous lower bounds, we disallow random accesses by working in the pointer machine
model of computation, however, we obtain high and unconditional space and query time lower bounds that almost match
the best known data structures for all the problems considered; at the moment, obtaining such 
unconditional and tight lower bounds in other models of computation is a hopelessly difficult problem. 
Furthermore, compared to the previous lower bounds in the area, our lower bounds probe deeper and thus are much more informative.
Consequently, these results show the usefulness of the pointer machine model.

Document indexing is an important problem in the field of information retrieval. 
Generally,  the input is 
a collection of documents $\mathcal{D} = \{d_1, d_2, \ldots d_D\}$ with  a total length of
$n$ characters and usually the goal is to index them such that 
given a query, all the documents matching the query can be either found efficiently (the \emph{reporting variant})
or counted efficiently (the \emph{searching variant}).
When the query is just one text pattern, the problem is classical
and well-studied and there are linear space solutions with optimal query time \cite{Muthukrishnan02}. 
Not surprisingly, there have been various natural extensions of this problem. 
We summarize the problems we study and our results below.

\subsection{New and Previous Results}

\newcommand{\doubleline}[2]{
	\begin{tabular}{@{}c@{}} 
		#1 \\
		\footnotesize{#2}
	\end{tabular}
}
\begin{table}[h]
	\begin{center}
        \begin{tabular}{l|l|l|l}
\textbf{Problem} 												& \textbf{Query Bound}  					& \textbf{Space Lower Bound}    &\textbf{Assumptions}  \\ \hline \hline
\doubleline{2P, FP, SI, 2FP}{(counting)}		& $Q(n)$ 								& $\Omega(n^{2-o(1)}/Q^2(n))$	& Semi-group \\ \hline
\doubleline{2P, FP, SI, 2FP}{(reporting)}		& $Q(n) + \O(t)$						& $\Omega(n^{2-o(1)}/Q(n))$		& PM \\ 
\doubleline{ 2P, FP, SI, 2FP}{(reporting)}	& $\O( (nt)^{\frac{1}{2}-\alpha} + t)$	& $\Omega\left(n^{\frac{1+6\alpha}{1+2\alpha}-o(1)}\right)$ &PM, $\alpha>0$ a parameter \\  \hline
WCI (reporting) 								& $Q(n,\kappa) + \O(t)$						& $\Omega\left(\frac{n}{\kappa} \Theta\left( \frac{\log n}{\kappa \log Q(n,\kappa)}  \right)^{\kappa-1}\right)$	& \doubleline{PM, $\kappa$ wild-cards,}{$\kappa \le \log_{Q(n,\kappa)}n$}  \\
WCI (reporting) 							& $\O(2^{\kappa/2} + t)$				& $\Omega\left(  n^{1+\Theta(1/\log k)}\right)$ 	& \doubleline{PM, $\kappa$ wild-cards,}{$3\sqrt{\log n} \le \kappa = o(\frac{\log n}{\log\log n})$} \\ \hline
$\kappa$-GPI									& $o\left( \frac{D\gamma^\kappa}{(2+2\log D)^{\kappa+1}} \right)$ & $n^{\Omega(\log^{1/(2\kappa)}n)}$	& \doubleline{PM, $\alpha_i = 0, \beta_i=\gamma$}{$\kappa = o(\frac{\log\log n}{\log\log\log n})$}\\
\end{tabular}
\end{center}
\caption{At every row, the 3rd cell presents our space lower bound for data structures that have a query time
    bounded by the 2nd cell. PM stands for the pointer machine model. $n$ is the input size and $t$ is the output size. }
\end{table}

\subsubsection{Two-pattern and the Related Queries}
The \emph{two-pattern query problem} (abbreviated as the 2P problem)
was introduced in 2001 by Ferragina et al.~\cite{Ferragina2Dind} and since
then it has attracted lots of attention. 
In the 2P problem, each query is composed of 
two patterns and a document matches the query if both patterns occur in the document.
One can also define the \emph{Forbidden Pattern} (FP) problem~\cite{FischerGKLMSV12} where
a document matches the query
if it contains one pattern  but not the other.
For symmetry, we also introduce and consider the \emph{Two Forbidden Patterns} (2FP) problem
where none of the patterns are required to match the document.
%
%Furthermore, in the searching variant, we can consider the documents to be 
%either weighted or unweighted; the weighted version is especially attractive in
%situations where the weight represents the ``importance'' of each document (e.g.\ 
%the PageRank value of a document).
%We will work with the weighted sets where a document
%$D_i$ is assigned a weight $w(D_i)$ from a semi-group $G$ (more details and motivations to follow) and
%we call this the \emph{semi-group variant}.
%%

\para{Previous Results.} 
Ferragina et al.~\cite{Ferragina2Dind} presented a number of solutions for the 2P problem
with space and query times that depend on the ``average size'' of
each document but the worst case query time and space is $\O(P_1 + P_2 + n^\alpha + t)$ and
$\O(n^{2-\alpha} \log^{\O(1)}n)$, for any $0 < \alpha < 1$, respectively.
Here $P_1$ and $P_2$ are the sizes of the query patterns and 
$t$ is the output size (see also \cite{Muthukrishnan02}).
Cohen and Porat~\cite{CohenP10} offered a solution that uses
$\O(n \log n)$ space with $\O(P_1 + P_2 +  \sqrt{n t} \log^{2.5}n)$ query time.
The space was improved to $\O(n)$ and the query time to $\O(P_1 + P_2 +  \sqrt{n t} \log^{1.5}n)$ by
Hon et al.~\cite{HonSTV10}. The query time  was reduced by a $\O(\sqrt{\log n})$ in \cite{LarsenMNT14} factor and finally
the query time $\O(P_1 + P_2 +  \sqrt{n t})$ was achieved in~\cite{BiswasGST15}.

The FP problem was introduced by Fischer et al.~\cite{FischerGKLMSV12} and they  presented a 
data structure that stores $\O(n^{3/2})$ bits and answers queries in  $\O(P_1 + P_2 + \sqrt n + t)$ time. 
Another solution was given by Hon et al.~\cite{HonSTV12} that uses $\O(n)$ space but has 
$\O(P_1 + P_2 +  \sqrt{n t} \log^{2.5}n)$ query time.
For the searching variant (unweighted) their solution can answer queries in
$\O(P_1 + P_2 +  \sqrt{n} \log\log n)$  time.
As Larsen et al.~\cite{LarsenMNT14} remark, the $\log\log n$ factor can be removed by 
using range emptiness data structures and that the same data structure can be used to
count the number of matching documents for the two pattern problem. 

The difficulty of obtaining fast data structures using (near) linear space
has led many to believe that very efficient solutions
are impossible to obtain.
Larsen et al.~\cite{LarsenMNT14} specifically focus on proving such impossibility claims
and they show that the 2P and FP problems are at least as hard as Boolean Matrix Multiplication,
meaning, with current techniques,
$P(n) + nQ(n) = \Omega(n^{\omega/2})$ where $P(n)$ and $Q(n)$
are the preprocessing and the query times of the data structure, respectively
and $\omega$ is the exponent of the best matrix multiplication algorithm (currently $\omega=2.3728639$).
If one assumes that there is no ``combinatorial'' matrix multiplication algorithm
with better running time than $O(n^{3-o(1)})$, then the lower bound becomes
$P(n) + nQ(n) = \Omega(n^{1.5 - o(1)})$.
Other conditional lower bounds for the 2P and FP problems but from the integer 3SUM conjecture
were obtained by Kopelowitz, et al.~\cite{KPPx,KPPc}.

The above results are conditional.
Furthermore, they tell us nothing about the
complexity of the space usage, $S(n)$,  versus the query time which is what we
are truly interested in for data structure problems. 
Furthermore, even under the relatively generous assumption\footnote{There are problems,
	such as jumbled indexing~\cite{ChanL15}, where the preprocessing time is a polynomial factor larger than the space complexity. } that $P(n) = \O(S(n) n^{o(1)})$,
	the space and query lower bounds obtained from the above results have polynomial gaps compared  with
	the current best data structures.

	We need to remark that the only unconditional space lower bound is a pointer machine lower bound
	that shows with query time of $\O(\textrm{poly}(\log n) + k)$ the space must be
	$\Omega(n(\log n / \log \log n)^3)$~\cite{FischerGKLMSV12}. 
	However this bound is very far away from the upper bounds (and also much lower than our lower bounds).

\para{Our Results.}
We show that all the known data structures for 2P and FP problems are optimal within $n^{o(1)}$ factors,
at least in the pointer machine model of computation:
Consider a pointer machine data structure that uses $S(n)$ space and can report all the $t$ documents that
match a given 2P query (or FP query, or 2FP query) in $Q(n) + \O(P_1 + P_2 + t)$ time.
We prove that we must have $S(n)Q(n) = \Omega\left(n^{2-o(1)}\right)$.
As a corollary of our lower bound construction, we also obtain that any data structure that
can answer 2P query (or FP query, or 2FP query) in
$\O( (nt)^{1/2 - \varepsilon} + t)$ time, for any fixed constant $\varepsilon > 0$ must use
super-linear space.
As a side result, we show that surprisingly, the counting variant of the problem is in fact easier:
in the semi-group model of computation (see \cite{Chazelle90b} or Appendix \ref{appendix:semi_group} 
for a description of the semi-group model), we prove that 
we must have 
%any data structure that solves the searching variant of the aforementioned problems with $S(n)$ space and $Q(n)$ query time must have 
$S(n)Q^2(n) = \Omega(n^2 / \log^{4}n)$.
%Furthermore, we show that both of the these trade off curves are almost optimal by providing data structures
%with space and query time trade-off that match our lower bounds, up to $n^{o(1)}$ factors.

\subsubsection{Set Intersection Queries.}
The interest in set intersection problems has grown considerably in recent years and variants of the set intersection problem
have appeared in many different contexts. 
Here, we work with the following variants of the problem.
The input is $m$ sets, $S_1, \cdots, S_m$ of total size $n$, from a universe $\U$ and 
queries are a pairs of indices $i$ and $j$.
The \emph{decision variant} asks whether $S_i \cap S_j = \emptyset$.
The \emph{reporting variant} asks for all the elements in $S_i \cap S_j$. 
In the \emph{counting variant}, the result should be $|S_i \cap S_j|$. 
In the \emph{searching variant}, the input also includes a weight function $w: \U \rightarrow G$ where 
$G$ is a semi-group. The query asks for $\sum_{x \in S_i \cap S_j} w(x)$.

\para{Previous Results.} The set intersection queries have appeared in many different formulations and
variants (e.g., see~\cite{CPoracle, PRoracle,PRToracle,popular,KPPx,CohenP10}).
The most prominent conjecture is that answering the decision variant 
with constant query time requires $\Omega(n^{2-o(1)})$ space~(see~\cite{PRoracle} for more details).
For the reporting variant,
Cohen and Porat~\cite{CohenP10} presented a data structure
that uses linear space and answers queries in $\O(\sqrt{nt})$ time, where $t$ is the output size.
They also presented a linear-space data structure for the searching variant that  answers queries in $\O(\sqrt{n})$ time. 
In~\cite{CTSEdycon} the authors study set intersection queries because of connections to  dynamic graph connectivity problems.
They offer very similar bounds to those offered by Cohen and Porat (with a $\sqrt{\log n}$ factor worse space and query times) 
but they allow updates in $\O(\sqrt{n \log n})$ time. 
It is commonly believed
that all set intersection queries are hard.  Explicitly stated conjectures on set intersection problems are used 
to obtain conditional lower bounds for problems such as
distance oracles~\cite{PRoracle, PRToracle, CPoracle} while other well-known conjectures, 
such as the 3SUM conjecture, can be used to show conditional lower bounds for variants of set intersection problems~\cite{KPPx,popular}.
For other variants see~\cite{fastunion,Pdylb,DietzMRU95}.

Dietz et al.~\cite{DietzMRU95} considered set intersection queries in the semi-group model
(a.k.a \textit{the arithmetic model}) and they presented near optimal dynamic and offline lower bounds.
They proved that given a sequence of $n$ updates and $q$ queries one must spend $\Omega(q + n\sqrt{q})$
time (ignoring polylog factors); in the offline version a sequence of $n$ insertions and $q$ queries
are used but in the dynamic version, the lower bound applies to a dynamic data structure that allows
insertion and deletion of points, as well as set intersection queries.

\para{Our Results.}
Our lower bounds for the 2P problem easily extend to the SI problem.
Perhaps the most interesting revelation here is that the searching variant is much easier than the reporting
variant ($S(n)Q(n) = \Omega(n^{2-o(1)})$ for reporting versus $S(n)Q(n)^2 = \Omega(n^{2-o(1)})$ for searching)
\footnote{ While we believe our lower bounds are certainly interesting
(particularly since they separate counting and reporting variants), they do not
make any progress towards resolving the status of the decision version which is
considered a major open problem.}.  Based on this, we make another conjecture
that even in the RAM model, reporting the elements in $S_i \cap S_j$ for two
given query indices $i$ and $j$, in $\O(n^{1-\varepsilon} + |S_i \cap S_j|)$
time, for any fixed constant $\varepsilon > 0$, requires $\omega(n)$ space.
Such a separation between counting and reporting is a rare phenomenon with
often counting being the difficult variant.

Observe that conditional lower bounds based on the Boolean Matrix Multiplication or the integer 3SUM 
conjectures have limitations in  distinguishing
between the counting and the reporting variants:
For example, consider the framework of Larsen et al.~\cite{LarsenMNT14} and for
the best outcome, assume $P(n) = S(n) n^{o(1)}$ and also that
boolean matrix multiplication requires $\Omega(n^{3-o(1)})$ time; 
then their framework  yields that $S(n) + n Q(n) = \Omega(n^{3/2-o(1)})$.  
When $Q(n) = \Theta(n^{2/3})$ this does not rule out the possibility of having $S(n) = \O(n)$
(in fact the counting variant \emph{can} be solved with linear space).
However, our lower bound shows that even with $Q(n)= \Theta(n^{2/3})$ the reporting
variant requires $\Omega(n^{4/3-o(1)})$ space.

\subsubsection{Wild Card Indexing}
We study the document indexing problem of matching with ``don't cares'', also
known as \emph{wild card matching} or \emph{indexing} (WCI).  The setup is the
same as for the 2P problem, except a query is a single pattern but it also
contains wild cards denoted by a special character ``$*$''. 
A ``$*$'' matches any \emph{one} character.  The task is
to report all documents matched by the query.  This is a well-studied problem
from the upper bound perspective and there are generally two variations: either
the maximum number of wild cards is bounded or it supports any number of wild
cards.  We consider the bounded version where patterns contain up to $\kappa$
wild cards and $\kappa$ is known in advance by the data structure. 

\para{Previous Results.}
Cole et al.~\cite{ColeGL04} presented a
data structure that uses $\O(n \frac{\log^\kappa n}{\kappa!})$ words of space and answers queries in $\O(P + 2^\kappa\log
\log n + t)$, where $t$ is the number of occurrences and $P$ is the length of the query pattern.
The space has been improved to $\O(n\log^{\kappa+\varepsilon}n)$ bits while keeping the same query time~\cite{lessspace}.
Another improvement came as a trade-off that increased query time to $\O(P + \beta^j \log \log n + t)$ and
reduced space usage $\O(n\log n\log_\beta^{\kappa-1} n)$ for any $2 \le \beta \le
\sigma$ where $\sigma$ is the alphabet size and $j \le \kappa$ is the number of wild cards in the query \cite{BilleGVV12}.  
In the same paper an alternate solution with $\O(P + t)$ query time and $\O(n\sigma^{\kappa^2}\log^\kappa n \log n)$ 
space usage was also presented.
Other results have focused on reducing space while increasing the query time but their query time now depends on the
alphabet size, e.g., in~\cite{wci14} the authors provide a data structure with $\O(n\log^{\varepsilon}n\log \sigma)$
space but with the query time of $\O(m + \sigma^\kappa\sqrt{\log\log\log n } + t)$.

From these bounds we note three things, first that all solutions have some exponential dependency on $\kappa$ and
second the alternate solution by Bille et al.~\cite{BilleGVV12} 
has an odd $\sigma^{\kappa^2}$ factor, which is exponential on a
\emph{quadratic} function of $\kappa$ as opposed to being exponential on a
linear function of $\kappa$ (such as $2^\kappa$, $\sigma^\kappa$, or
$\log^\kappa n$).
Third, when the query time is forced to be independent of $\sigma$, 
there is a discrepancy between query and space when varying $\kappa$:
Increasing $\kappa$ (when it is small) by one increases the space by approximately a $\log n$ factor,
while to increase the query time by a $\log n$ factor, $\kappa$ needs to be increased by $\log \log n$.
Based on the third point, it is quite likely that unlike SI or 2P problems, WCI does not have a simple trade-off curve.

Other results that are not directly comparable to ours include the following:
One is an $\O(n)$ space index with $\O(P + \alpha)$ query time \cite{RahmanI07}
where $\alpha$ is the number of occurrences of all the subpatterns separated by
wild card characters. 
Note that $\alpha$ could be much larger than $t$ and in fact, this can result in a worst case linear query time,
even with small values of $t$. 
Nonetheless, it could perform reasonably well in practice.
Two, there are lower bounds for the \emph{partial match} problem, which is a related problem
(see \cite{Patrascu11} for more details).

\para{Our Results.}
For WCI with $\kappa$ wild cards, we prove two results, both in the pointer machine model.
In summary,  we show that the exponential dependency of space complexity or query time on $\kappa$ 
generally cannot be avoided.

As our first result and for a binary alphabet ($\sigma=2$), 
we prove that for $3\sqrt{\log n} \le \kappa = o(\frac{\log n}{\log\log n})$, any data
structure that answers queries in $\O(2^{\kappa/2} + P + t)$ time must consume
$n^{1+\Theta(1/\log \kappa)}$ space. 
This result rules out the possibility of lowering the $\sigma^{\kappa^2}$ factor in the alternate solution
offered by Bille et al.~\cite{BilleGVV12}, over all values of $\kappa$, to $\sigma^{\kappa^{2-\varepsilon}}$
for any constant $\varepsilon > 0$:
by setting $\kappa= 3\sqrt{\log n}$ (and
$\sigma=2$), such a bound will be much smaller than our  space lower bound
(essentially involves comparing $2^{\O(\log^{1-\varepsilon/2}n)}$ factor to
 a  $2^{\Omega(\log n / \log\log n)}$ factor).  While this does not rule out
improving the space bound for small values of $\kappa$, it shows that the
exponential dependency on $\kappa^2$ is almost tight at least in a
particular point in the range of parameters.

As our second result, we prove that answering WCI queries in
$Q(n,\kappa) + \O(P+t)$ time requires
$\Omega( \frac{n}{\kappa} \Theta\left(\frac{ \log_{Q(n,\kappa)} n}{\kappa}\right)^{\kappa-1})$ space,
as long as $\kappa < \log_{Q(n,\kappa)} n$.
Note that this query time is assumed to be independent of $\sigma$.
This result also has a number of consequences.
One, it shows that 
any data structure with query time of $\O( \log^{\O(1)} n + P + t)$ requires 
$\Omega(\frac{n}{\kappa} \left( \frac{\log n}{\kappa\log\log n}\right)^{\kappa-1})$ space.
Note that this is rather close to the space complexity of the data structure of 
Cole et al.~\cite{ColeGL04} ($2^{\O(\kappa)} (\log\log n)^\kappa$ factor away).
In other words, the data structure of Cole et al. cannot be significantly improved both in space complexity and query time;
e.g., it is impossible to answer WCI queries in $\O( \log^{\O(1)} n + P + t)$ time
using $\O\left(n \left(\frac{\log n}{\kappa})^{\kappa(1-\varepsilon)} \right)\right)$ space, for any
constant $\varepsilon > 0$.

Two, $\kappa=2$ is the smallest value where
linear space becomes impossible with polylogarithmic query time.
This is very nice since $\kappa=1$ can be solved with almost linear space~\cite{lessspace}.
Furthermore this shows that the increase by a $\log n$ factor in the space complexity
for every increase of $\kappa$ is necessary (for small values of $\kappa$).

Three, we can combine our second result with our first result when $\kappa = 3 \sqrt{\log n}$.
As discussed before, our first result rules out fast queries (e.g., when $Q(n) \le 2^{\kappa/2}$),
unless the data structure uses large amounts of space, so consider the case when $Q(n) = \O(2^\kappa)$. 
In this case, we can rule out  lowering the space usage of the data structure of Cole et al. to 
$\Omega\left( n \left( \frac{\log n}{\kappa}  \right)^{\kappa^{1-\varepsilon}}\right)$ for any
constant $\varepsilon > 0$:
apply our second lower bound 
with fewer wild cards, specifically, with $\kappa' = \kappa^{1-\delta}$ wild cards, for a small enough constant $\delta > 0$
that depends on $\varepsilon$.
Observe that $\kappa' < \log_{Q(n)} n$, so the second result lower bounds the space by 
$\Omega\left( \frac{n}{\kappa'} \Theta\left( \log^{ \frac{\delta}{2} \kappa^{1-\delta}} n \right)\right)$,
which for a sufficiently small $\delta$ is greater than~$\Omega\left( n \left( \frac{\log n}{\kappa}  \right)^{\kappa^{1-\varepsilon}}\right)$.

As mentioned in the beginning, our results show that the exponential dependency of space or query time on $\kappa$ cannot be improved in general.
Furthermore, at a particular point in the range of parameters (when $\kappa = 3 \sqrt{\log n}$), 
all the known exponential dependencies on $\kappa$ are almost tight and cannot be lowered to an exponential dependency on $\kappa^{1-\varepsilon}$
(or on $\kappa^{2-\varepsilon}$ for the alternate solution) for any constant $\varepsilon > 0$.
Nonetheless, there are still gaps between our lower bounds and the known data structures.
We believe it is quite likely both our lower bounds and the existing data structures can be improved to narrow the gap. 

\subsubsection{Gapped Pattern Indexing}
A $\kappa$-gapped pattern is a pattern $p_1\{\alpha_1, \beta_1\}
p_2\{\alpha_2,\beta_2\}, \cdots, p_\kappa\{\alpha_\kappa,\beta_\kappa\}p_{\kappa+1}$
where $\alpha_i$ and $\beta_i$, $1 \le i \le \kappa$ are integers, and each $p_i$, $1 \le i \le \kappa+1$, 
is a string over an alphabet of size $\sigma$.
Such a $\kappa$-gapped pattern matches a documents in which one can find one occurrence of every $p_i$ such that
there are at least $\beta_i$ and at most $\alpha_i$ characters between the occurrence of $p_i$ and the occurrence of 
$p_{i+1}$, $1 \le i \le \kappa$.

\para{Previous Results.}
The gapped pattern indexing is often considered both in the online and the offline version 
(e.g., see~\cite{onegap,unevengaps}).
However, the result most relevant to us is \cite{mindthegap}, where they consider the following data structure problem:
given a set of $1$-gapped patterns of total size $n$, where all the patterns are in the form of
$p_1\{\alpha, \beta\}p_2$, store them in a data structure such that
given a document of length $D$ at the query time, 
one can find all the gapped patterns that match the query document (in general
we call this the \emph{$\kappa$-gapped pattern indexing ($\kappa$-GPI)} problem).
In~\cite{mindthegap}, the authors give a number of upper bounds and conditional lower bounds for the problem.
Among a number of results, they can build a data structure of linear size that can answer queries in
$\tilde{\O}(D (\beta-\alpha) + t)$ where $\tilde{\O}$ notation hides polylogarithmic factors and
$t$ is the output size.
For the lower bound and among a number of results, 
they can show that with linear space $\Omega(D (\beta-\alpha)^{1-o(1)} + t)$ query time is
needed.

\para{Our Results.}
We consider the general $\kappa$-GPI problem where $\beta_i -\alpha_i = \gamma$
for all $1 \le i \le \kappa$, and a prove lower bound that is surprisingly very high: 
any pointer machine data structure that can answer queries in 
$o( D \gamma^{\kappa} /( 2\log D)^{\kappa+1})$ time must use \emph{super polynomial} space
of $n^{\Omega(\log^{1/(\kappa+1)}n)}$.
By construction, this result also holds if the input is a set of $\kappa+1$ patterns
where they all need to match the query document (regardless of their order and size of the gaps).
In this case, answering queries in $o( D^{\kappa+1} /( 2\log D)^{\kappa+1})$ requires the same
super-polynomial space.
Note that in this case $\kappa=1$ is the ``dual'' of the 2P problem:  store a set of two-patterns in
data structure such that given a query document, we can output the subset of two-patterns that
match the document.

\subsection{Technical Preliminaries}

\para{The Pointer Machine Model~\cite{TarjanPM}.}
This models data structures that solely use
pointers to access memory locations (e.g., any tree-based data structure)\footnote{
    Many of the known solutions for various indexing problems use tree structures, such as suffix
trees or wavelet trees.
While sometimes trees are can be encoded using bit-vectors with rank/select structures on top,  
these tricks can only save polylogarithmic factors in space and query times.}.
We focus on a variant that is the popular 
choice when proving lower bounds~\cite{Chazelle90a}.
Consider an abstract ``reporting'' problem where the input includes a universe set 
$\U$ where each query $q$ reports a subset, $q_\U$, of $\U$. 
The data structure is modelled as a directed graph $G$ 
with outdegree two (and a root node $r(G)$) where each vertex represents
one memory cell and each memory cell can store one element of $\U$; 
edges between vertices represent pointers between the memory cells. 
All other information
can be stored and accessed for free by the data structure. The only requirement is that given 
the query $q$, the data structure must start at $r(G)$ and explore a connected subgraph of $G$ and find its way 
to  vertices of $G$ that store the elements of $q_\U$. 
The size of the subgraph explored is a lower bound on the query time and the size of $G$ is a lower bound on the
space.

\para{An important remark.} The pointer-machine can be used to prove lower bounds for data structures with query time
$Q(n) + \O(t)$ where $t$ is the output size and $Q(n)$ is ``search overhead''. 
Since we can simulate any RAM algorithm on a pointer-machine with $\log n$ factor slow down,
we cannot hope to get high unconditional lower bounds if we assume the query time is $Q(n) + \O(t \log n)$,
since that would automatically imply RAM lower bounds for data structures with 
$Q(n)/\log n + \O(t)$ query time, something that is hopelessly impossible with current techniques.
However, when restricted to query time of $Q(n) + \O(t)$, the pointer-machine model is 
an attractive choice and it has an impressive track record of proving lower bounds that match the best 
known data structures up to very small factors, even when compared to RAM data structures; we mention
two prominent examples here.
For the fundamental \textit{simplex range reporting} problem, all known solutions are pointer-machine
data structures~\cite{Matousek.RS.hierarchal.93,Chan.ParTree,Chazelle.et.al.simplex.RS} and the
known pointer machine lower bounds match these up to an $n^{o(1)}$ factor~\cite{Afshani12,ChazelleR95}.
One can argue that it is difficult to use the power of RAM for simplex range reporting problem.
However, for the other fundamental \textit{orthogonal range reporting}, where it is easy to do
various RAM tricks, the best RAM data structures save at most a $\log n$ factor compared to the best
known pointer machine solutions~(e.g., see \cite{aal10,aal12,clp11}).
Also, where cell-probe lower bounds cannot break the $\log n$ query barrier, 
very high lower bounds are known for the orthogonal range reporting problem in the pointer machine
model~\cite{aal10,aal12,Chazelle90a}. 

\para{Known Frameworks.}
The fundamental limitation in the pointer machine model is that starting from a memory cell $v$,
one can visit at most $2^{\ell}$ other memory cells using $\ell$ pointer navigations.
There are two known methods that exploit this limitation and build two different frameworks for
proving lower bounds. 

The first lower bound framework was given by Bernard Chazelle~\cite{Chazelle90a,ChazelleR95}.
However, we will need a slightly improved version of his framework that is presented in the following lemma;
essentially, we need a slightly tighter analysis on a parameter that was originally intended as a large constant.
Due to lack of space we defer the proof to the Appendix~\ref{sec:chazelle}.

\begin{restatable}{theorem}{thmlb}\label{thm:lblemma}
Let $\mathcal{U}$ be a set of $n$ input elements and $\Q$ a set of queries where each query outputs
a subset of $\mathcal{U}$. Assume there exists a data structure that uses $S(n)$ space
and answers each query in $Q(n) + \alpha k$ time, where $k$ is the output size.
Assume (i) the output size of any query $q \in \Q$, denoted by $|\U \cap q|$, is at least $t$, for a parameter $t \ge Q(n)$
and (ii) for integers $\ell$ and $\beta$, and indices, $i_1, \cdots, i_\ell$,
  $|\U~\cap~q_{i_1}~\cap~\cdots~\cap~q_{i_{\ell}}|~<~\beta$.  Then,
   $ S(n) = \Omega\left(\frac{|\Q|t}{\ell \cdot 2^{O(\alpha \beta)}}\right)$.
\end{restatable}

The second framework is due to Afshani~\cite{Afshani12} and it is designed  for 
``geometric stabbing problems'': given an input set of $n$ geometric regions, the goal is store them
in a data structure such that given a query point $q$, one can output the subset of regions that contain $q$.
The framework is summarized below.
\begin{theorem}\cite{Afshani12}\label{thm:dual}
	Assume one can construct a set of $n$ geometric regions inside the $d$-dimensional 
	unit cube such that (i) every point of the unit cube is contained in at least $t$
	regions \footnote{In~\cite{Afshani12} this is
stated as ``exactly $t$ ranges'' but the proof works with only a lower bound on $t$.}, and (ii) the volume of the intersection of every $\beta$ regions is
	at most $v$, for some parameters $\beta$, $t$, and $v$. 
	Then, for any pointer-machine data structure that uses $S(n)$ space and can answer geometric stabbing queries on the above input
    in time $g(n) + O(k)$, where $k$ is the output size and $g(\cdot)$ is some
    increasing function, if $g(n) \le t$ then $S(n) = \Omega( t
    v^{-1}2^{-O(\beta)})$.
\end{theorem}

These two frameworks are not easily comparable. 
In fact, for many constructions, often only one of them gives a non-trivial lower bound.
Furthermore, as remarked by Afshani~\cite{Afshani12}, Theorem~\ref{thm:dual} does not need to be operated in the
$d$-dimensional unit cube and in fact any measure could be substituted instead of the $d$-dimensional
Lebesgue measure.

\para{Our Techniques.}
Our first technical contribution is to use Theorem~\ref{thm:dual} in a non-geometric setting by 
representing queries as abstract points under a discrete measure and each input object as a range
that contains all the matching query points.
Our lower bound for the $\kappa$-GPI problem and one of our WCI lower bounds are proved in this way.
The second technical contribution is actually building and analyzing proper input and query sets to be used in the 
lower bound frameworks.
In general, this is not easy and in fact for some problems it is highly challenging\footnote{
A good example is the classical halfspace range reporting problem where constructing
 proper input and query sets has been a longstanding open problem; the current best lower bound
 uses a highly inefficient reduction from the simple range reporting problem~\cite{Afshani12}.}. 
Also see Section~\ref{sec:conc} (Conclusions) for a list of open problems.

In the rest of this article, we present the technical details behind
our $\kappa$-GPI lower bound and most of the details of our 
first WCI lower bound. 
Due to lack of space, the rest of the technical details have been moved to the appendix.

We begin with the $\kappa$-GPI problem since it turns out for this particular problem we can
get away with a simple deterministic construction. 
For WCI, we need a more complicated randomized construction to get the best result and thus it is presented
next.

%\todo{Make it absolutely clear that our contribution is not yet another application of pointer machine lower bound but is in the design of hard input instances for string problems and proving they are in fact hard}

%% What exactly does the pointer machine lower bound say? It gives a
%% lower bound on the number of pointers one must follow (query time)
%% versus the number of pointers one must store. If one assumes that
%% documents are ``indivisible'', in the sense that one cannot mangle the
%% bits of the document identifier then the lower bound should still hold
%% within polylogarithmic factors. Argue that current upper bounds do not
%% mangle the bits of the document identifiers. Essentially the pointer
%% machine lower bounds says one must store the document id in many
%% places.

%% file: gapped_pattern.tex
\section{Gapped Pattern Lower Bound}
\label{sec:gapped}

In this section we deal with the data structure version of the gapped pattern problem.
The input is a collection of $\kappa$-gapped patterns (typically called a
dictionary), with total length $n$ (in characters).  The goal is to store the
input in a data structure such that given a document of size $D$, one can
report all the input gapped patterns that match the query.  
We focus on special $\kappa$-gapped patterns that we call \emph{standard}: 
a standard $\kappa$-gapped pattern in the form of $p_1\left\{ 0,\gamma
\right\}p_2\left\{ 0,\gamma \right\}\dots\left\{ 0,\gamma
\right\}p_{\kappa+1}$ where each $p_i$ is a string (which we call a \emph{subpattern}) and $\gamma$ is an integer.

\begin{restatable}{theorem}{thmgapped}\label{thm:gapped}
    For $\kappa = o(\frac{\log\log n}{\log\log\log n})$ and in the pointer machine model,
    answering $\kappa$-GPI queries in
  $o\left(\frac{D\gamma^{\kappa}}{(2+2\log D)^{\kappa+1}}\right) + O(t)$ time
  requires $n^{\Omega(\log^{1/(2\kappa)} n)}$ space.
\end{restatable}

To prove this theorem, we build a particular input set of standard $\kappa$-gapped patterns.
We pick the alphabet $\Sigma = \{0,1,\#\}$, and the gapped patterns only use $\{0,1\}$.
Each subpattern in the input is a binary string of length $p$.
The subpatterns in any gapped pattern are in lexicographic order, and a subpattern occurs at most once in a pattern
(i.e., no two subpatterns in a pattern are identical).
The input set, $S$, contains all possible gapped patterns obeying these conditions.
Thus, $|S| = {2^p \choose \kappa+1}$.
Each query is a text composed of concatenation of $D/(p+1)$ substrings (for simplicity, we assume $D/(p+1)$ is an integer)
and each substring is in the form '$\#\{0,1\}^p$'.
We restrict the query substrings to be in lexicographic order
and without repetitions (no two substrings in a query are identical).
The set of all query texts satisfying these constraints is denoted by~$Q$.
Thus, $|Q| = {2^p \choose D/(p+1)}$.

\begin{lemma}
  \label{lemma:output}
  For $D \ge 2\kappa\gamma$, any text $T \in Q$ matches $\Theta\left(\tfrac{D\gamma^{\kappa}}{(p+1)^{\kappa+1}}\right)$
  gapped patterns in $S$.
\end{lemma}
\begin{proof}
	Consider a text $T \in Q$. To count the number of gapped patterns that can match it,
	we count the different ways of selecting $\kappa+1$ positions that correspond to the starting positions
	of a matching subpattern. 
	Each position starts immediately after '\#', with at most $\gamma$ characters between consecutive positions.
  Since $D \ge 2\kappa\gamma$, we have $\Theta(D/p)$ choices for picking the first position, i.e., 
  the starting position of a gapped pattern matching $T$.
  After fixing the first match, there are at most $\gamma/(p+1)$ choices for the position of the next match
  between a subpattern and substring.
  However, if the first match happens in the first half of text $T$, there are always $\gamma/(p+1)$ choices
  for the position of each subpattern match (since $D \ge 2\kappa\gamma$).
  Thus, we have $\Theta\left( D\gamma^k / (p+1)^{\kappa+1} \right)$ choices. As input subpatterns
  are in lexicographically increasing order, different choices result in distinct gapped patterns
  that match the query.
\end{proof}

To apply Theorem~\ref{thm:dual}, we consider each query text $T$ in $Q$ as a ``discrete'' point
with measure $1/|Q|$. 
Thus, the total measure of $Q$ (i.e., the query space) is one and $Q$ functions as the ``unit cube'' within
the framework of Theorem~\ref{thm:dual}.
We consider an input gapped pattern $P$ in $S$ as a range that contains all the points of $Q$ that match $P$. 
Thus, to apply Theorem~\ref{thm:dual}, we need to find a lower bound on the output size of every query (condition
(i)) and an upper bound on $v$, the measure of the intersection of $\beta$ inputs (condition (ii)). 
By the above lemma, meeting the first condition is quite easy:
we pick $t=\Theta\left(\tfrac{D\gamma^{\kappa}}{(p+1)^{\kappa+1}}\right)$ (with the right hidden constant).
Later we shall see that $p = 1+2\log D$ so this can be written 
as  $ t = \Theta\left(\tfrac{D\gamma^{\kappa}}{(2+2\log D)^{\kappa+1}}\right)$
Thus, we only need to upper bound $v$  which we do below.

\begin{lemma}
  \label{lemma:intersect}
  Consider $\beta$ patterns  $P_1, \cdots, P_\beta \in S$.
  At most 
  $\binom{2^p-\beta^{1/(\kappa+1)}}{D/(p+1)-\beta^{1/(\kappa+1)}}$
  texts in $Q$ can match all patterns $P_1, \cdots, P_\beta$.
\end{lemma}
\begin{proof}
    Collectively, $P_1, \cdots, P_\beta$ must contain at least $r=\beta^{1/(\kappa+1)}$ distinct subpatterns:
    otherwise, we can form at most ${r
    \choose \kappa+1} < \beta$ different gapped patterns, a contradiction.
    This in turn implies that any text $T$ matching $P_1, \cdots, P_\beta$ must contain all these at least $r$
    distinct subpatterns. 
    Clearly, the number of such texts is at most 
  $\binom{2^p-\beta^{1/(\kappa+1)}}{D/(p+1)-\beta^{1/(\kappa+1)}}$.
\end{proof}

As the measure of each query in $Q$ is $1/|Q|$, by the above theorem, we have
$v \le \binom{2^p-\beta^{1/(\kappa+1)}}{D/(p+1)-\beta^{1/(\kappa+1)}}/|Q|$.
We now can apply Theorem~\ref{thm:dual}.
Each pattern in the input has $\Theta(p\kappa)$ characters and thus the total input size, $n$, is
$\Theta(p \kappa|S| ) = \Theta(p\kappa {\binom{2^p}{\kappa+1}})$.
By the framework, and Lemmata~\ref{lemma:intersect} and \ref{lemma:output}, we know that the space
usage of the data structure is at least 
\vspace{-1mm}
\begin{eqnarray*}
    \Omega\left(\tfrac{D\gamma^{\kappa}}{(p+1)^{\kappa+1}} \cdot \frac{\binom{2^p}{D/(p+1)}}{\binom{2^p-\beta^{1/(\kappa+1)}}{D/(p+1)-\beta^{1/(\kappa+1)}}} \cdot 2^{-\O(\beta)}\right)
    = \Omega\left( 1 \cdot
  \left(\frac{(p+1)2^{(p-1)}}{D}\right)^{\beta^{1/(\kappa+1)}} \cdot
  2^{-\O(\beta)} \right)
\end{eqnarray*}
%% \begin{eqnarray*}
%%   S(n) &=& \Omega\left(\tfrac{t\gamma^{s-1}}{p^s} \cdot \frac{\binom{2^p}{t/p}}{\binom{2^p-\beta^{1/s}}{t/p-\beta^{1/s}}} \cdot 2^{-\O(\beta)}\right)
%%   = \Omega\left(\tfrac{t\gamma^{s-1}}{p^s} \cdot
%%   \frac{2^p!(t/p-\beta^{1/s})!(2^p-t/p)!}{(t/p)!(2^p-t/p)!(2^p-\beta^{1/s})!} \cdot
%%   2^{-\O(\beta)}\right)\\
%%   &=& \Omega\left(\tfrac{t\gamma^{s-1}}{p^s} \cdot
%%   \frac{2^p!(t/p-\beta^{1/s})!}{(t/p)!(2^p-\beta^{1/s})!} \cdot
%%   2^{-\O(\beta)}\right)
%%   = \Omega\left(\frac{t\gamma^{s-1}}{p^s} \cdot
%%   \frac{2^p!}{(2^p-\beta^{1/s})!} \cdot \left(\frac{p}{t}\right)^{\beta^{1/s}} \cdot
%%   2^{-\O(\beta)} \right)\\
%%   &=& \Omega\left(\frac{t\gamma^{s-1}}{p^s} \cdot
%%   2^{(p-1)\beta^{1/s}} \cdot \left(\frac{p}{t}\right)^{\beta^{1/s}} \cdot
%%   2^{-\O(\beta)} \right)
%%   = \Omega\left(\frac{t\gamma^{s-1}}{p^s} \cdot
%%   \left(\frac{p2^{(p-1)}}{t}\right)^{\beta^{1/s}} \cdot
%%   2^{-\O(\beta)} \right)
%% \end{eqnarray*}
where to obtain the rightmost equation we expand the binomials, simplify,
and constrain $\beta^{1/(\kappa+1)} < 2^p/2$ to lower bound $2^p-\beta^{1/(\kappa+1)}$ with $2^{p-1}$.
Now we work out the parameters. 
We know that $n = \Theta(p\kappa \binom{2^p}{\kappa+1}) = 2^{\Theta(p(\kappa+1))}$; this is satisfied by 
setting $p = c_p (\log n) / (\kappa+1)$ for some constant $c_p$.
Observe that there is an implicit constraint on $D$ and $p$:
there should be sufficient bits in the subpatterns to fill out a query with distinct subpatterns, i.e. $2^p > D/(p+1)$;
we pick $D =  n^{c_D 1/(\kappa+1)}$ for some other constant $c_D$ such that $D = 2^{p/2-1}$
and thus $(p+1)2^{p-1}/ D = 2^{p/2}$.
Using these values, the space lower bound is simplified to
\[ \Omega\left(  2^{\beta^{1/(\kappa+1)} \frac{c_p}{2}\frac{\log n}{\kappa+1} }\cdot 2^{-c_\beta \beta}\right) \]
where $c_\beta$ is another constant.
We now optimize the lower bound by picking $\beta$ such that
$c_\beta \beta = \frac{1}{2}\beta^{1/(\kappa+1)} \frac{c_p}{2}\frac{\log n}{\kappa+1}$
which solves to $\beta = \Theta( (\frac{\log n}{\kappa+1})^{1+1/\kappa})$.
Thus, for  constant $c$, the space complexity of the data structure is
\[ \Omega\left( 2^{c \left(\frac{\log n}{\kappa+1}\right)^{1+\frac{1}{\kappa}}}\right) =  
\Omega\left( 2^{\log n \cdot c \log^{\frac{1}{\kappa}}n \left(\frac{1}{\kappa+1}\right)^{1+\frac{1}{\kappa}}}\right) =  
n^{\Omega\left(  \log^{\frac{1}{\kappa}}n \left(\frac{1}{\kappa+1}\right)^{1+\frac{1}{\kappa}} \right)} = 
n^{\Omega\left(  \log^{\frac{1}{2\kappa}}n \right)}\]
where the last part follows from $\kappa = o( \log\log n / \log\log\log n)$.

%% file: query_lb.tex
\section{Wild Card Indexing}
In this section we consider the wild card indexing (WCI) problem and prove both space and query lower bounds
in the pointer machine model of computation.
Note that our query lower bound applies to an alphabet size of two (i.e., binary strings).

\subsection{The Query Lower Bound}\label{subsec:querywci}
Assume for any input set of
documents of total size $n$, we can build a data structure such that given a WCI query of length
$m$ containing $\kappa$ wild cards, we can find all the documents that match the query in 
$Q(n,\kappa) + O(m + t)$ time, where $t$ is the output size. 
Furthermore, assume $\kappa$ is a fixed parameter known by the data structure and that
the data structure consumes $S(n,\kappa)$ space. 
Our main result here is the following.

\begin{theorem}\label{thm:wciq}
	If $3 \sqrt{\log n} \le \kappa = o(\log n)$ and
	$Q(n,\kappa) = O(2^{\frac{\kappa}{2}})$, then $S(n,\kappa) \ge n^{1+\Theta(\frac{1}{\log \kappa})}$.
\end{theorem}

To prove the lower bound, we build a \textit{particular} set of documents and patterns and prove
that if the data structure can answer the queries fast, then it must consume lots of space, 
for this particular input, meaning, we get lower bounds for the function $S(\cdot)$.
We now present the details. We assume $Q(n,\kappa) \le  2^{\kappa/2}$, as otherwise the theorem is trivial.

\para{Documents and patterns.}
We build the set of documents in two stages.
Consider the set of all bit strings of length $m$ with exactly $\ell=\kappa/2$ ``1''s. 
In the first stage, we sample each such string independently with probability $r^{-1}$ where $r=2^{\kappa/3}$.
Let $\mathcal{D}$ be the set of sampled strings.
In the second stage,
for every set of $\ell+\ell'$ indices, $1 \le i_1 < i_2 < \cdots < i_{\ell+\ell'} \le m$, where
$\ell' = (\log_\ell r)/2 = \Theta(\kappa/\log \kappa)$, we  perform the following operation, given another parameter $\beta$:
if there are more than $\beta$ strings in $\mathcal{D}$ that have ``1''s only among positions
$i_1, \cdots, i_{\ell+\ell'}$, then we remove all such strings from $\mathcal{D}$.
Consequently, among the remaining strings in $\mathcal{D}$, ``1''s in every subset of $\beta$ strings 
will be spread over at least $\ell+\ell'+1$ positions.
The set of remaining strings $\mathcal{D}$ will form our input set of documents.
Now we consider the set $\P$ of all the patterns of length $m$ that have exactly $\kappa$ wild cards and 
$m-\kappa$ ``0''s.
We remove from $\P$ any pattern that matches fewer than ${\kappa\choose \ell}/(2r)$ documents from $\mathcal{D}$.
The remaining patterns in $\P$ will form our query set of patterns.
In Appendix~\ref{sec:nice}, we prove the following.
\begin{restatable}{lemma}{lemnicel}\label{lem:nice}
	With positive probability, we get a set $\mathcal{D}$ of $\Theta( {m \choose \ell} /r)$ documents
	and a set $\P$ of $\Theta( {m \choose \kappa})$ patterns such that
	(i) each pattern matches $\Theta( {\kappa \choose \ell}/r)$ documents, and
	(ii) there are no $\beta=\Theta(\log_\kappa m)$ documents whose ``1''s are contained
	in a set of $\ell + \ell'$ indices.
\end{restatable}
\toappendix{\lemnice}{\true}{\section{Proof of Lemma~\ref{lem:nice}}\label{sec:nice}\lemnicel*}{
	\begin{proof}
			Observe that the second stage of our construction guarantees property (ii).
			So it remains to prove the rest of the claims in the lemma.

		Consider a sampled document (bit string) $d$.
		Conditioned on the probability that $d$ has been sampled, 
		we compute the probability that $d$ gets removed in the second stage of our sampling
		(due to conflict with $\beta-1$ other sampled documents).

		Let $I \subset [m]$ be the set of indices that describe the position of ``1''s in $d$.
		We know that $|I| = \ell$ by construction.
		Consider a subset $I' \subset [m]\setminus I$ with $|I'| = \ell'$.
		By our construction, we know that if there are $\beta-1$ other sampled documents 
		whose ``1''s are at positions $I\cup I'$, then we will remove $d$ (as well as all those
		$\beta-1$ documents).
		We first compute the probability that this happens for a fixed set $I'$ and then 
		use the union bound.
		The total number of documents that have ``1''s in positions $I\cup I'$ is
		\begin{align}
				{\ell + \ell' \choose \ell'} \le \left( \frac{\ell e}{\ell'} \right)^{\ell'} < \ell^{\ell'}\le \sqrt{r}
		\end{align}
		and thus we expect to sample at most  $1/ \sqrt{r}$	 of them.
		We bound the probability that instead $\beta$ documents among them are sampled.
		We use the Chernoff bound by picking 
		$\mu = {1}/ \sqrt{r}$, and $(1+\delta)\mu = \beta$ and
		we obtain that the probability that $\beta$ of these documents are
		sampled is bounded by
		\begin{align*}
				\left(\frac{ e^\delta}{(1+\delta)^{1+\delta}}  \right)^\mu \le
				\left(\frac{O(1)}{\sqrt{r} \beta}  \right)^\beta.
		\end{align*}
		We use the union bound now. The total number of possible ways for picking the set $I'$
		is ${m-\ell \choose \ell'} {\ell \choose \ell'} < m^{2\ell'}$
		which means the probability that we remove document $d$ at the second stage of our construction
		is less than
		\begin{align*}
				\left(\frac{O(1)}{\sqrt{r} \beta}  \right)^\beta m^{2\ell'} < \frac{1}{r^6} = 2^{-2\kappa}
		\end{align*}
		if we pick $\beta = c \log_\kappa m$ for a large enough constant $c$.
		Thus, most documents are expected to survive the second stage of our construction.
		Now we turn our attention to the patterns.

		For a fixed pattern $p \in \P$, there are ${\kappa \choose \ell} = {\kappa \choose \kappa/2} \ge 2^\kappa/\kappa$ documents that could
		match $p$ and among them we expect to sample  ${\kappa \choose \ell}/r$ documents.
		An easy application of the Chernoff bound can prove that with high probability, we will sample
		a constant fraction of this expected value, for every pattern, in the first stage of our construction.
		Since the probability of removing every document in the second stage is at most $2^{-2\kappa}$, 
		the probability that a pattern is removed at the second stage is
		less than $2^{2\kappa/3}/(2\kappa) 2^{-2\kappa} < 2^{-\kappa}$ and thus, we expect a constant
		fraction of them to survive the second stage.
	\end{proof}
}

To prove the  lower bound, we use Theorem~\ref{thm:dual}.
We use a discrete measure here:
each pattern is modelled as a ``discrete'' point with measure $\frac{1}{|\P|}$, meaning,
the space of all patterns has measure one. 
Each document forms a range: a document $d_i$ contains all the patterns (discrete points) that match $d_i$. 
Thus, the measure of every document $d_i \in \mathcal{D}$ is $t_i/|\P|$, where $t_i$ is the number of patterns that match $d_i$.
We consider the measure of the intersection of $\beta$ documents (regions) $d_1, \dots, d_\beta$.
By Lemma~\ref{lem:nice}, there are $\ell + \ell'$ indices where one of these documents has a ``1'';
any pattern that matches all of these documents must have a wild card in all of those positions.
This means, there are at most $m-\ell - \ell' \choose  \kappa - \ell - \ell'$
patterns that could match documents $d_1, \dots, d_\beta$.
This means, when we consider documents as ranges, the intersection of every $\beta$ documents
has measure at most ${m-\ell - \ell' \choose  \kappa - \ell - \ell'} / |\P|$ which is an upper bound
for parameter $v$ in Theorem~\ref{thm:dual}.
For the two other parameters $t$ and $g(n)$ in the theorem we have,
$t = \Theta( {\kappa \choose \ell}/r)$ (by Lemma~\ref{lem:nice}) and
$g(n) = Q(n,\kappa) + O(m)$.
To obtain a space lower bound from Theorem~\ref{thm:dual}, we must check
if $g(n) \le t = \Theta( {\kappa \choose \ell}/r)$.  Observe that
${\kappa\choose \ell} = {\kappa\choose \kappa/2} \ge 2^\kappa/\kappa$ since
the binomials $\kappa \choose i$ sum to $2^\kappa$ for $0 \le i \le \kappa$ and
$\kappa \choose \kappa/2$ is the largest one. 
As $r= 2^{-\kappa/3}$, we have $t = \Omega(2^\kappa 2^{-\kappa/3}/\kappa) = \omega(2^{\kappa/2}) = \omega(Q(n,k))$.
However, $g(n)$ also involves an additive $O(m)$ term. 
Thus, we must also have $t = \omega(m)$ which will hold for our choice of parameters but we will verify it later. 

By guaranteeing that $g(n) \le t$, Theorem~\ref{thm:dual} gives a space lower bound of 
$\Omega( t v^{-1} 2^{-O(\beta)})$.
However, we would like to create an input of size $\Theta(n)$ which means the number of sampled
documents must be $\Theta(n/m)$ and thus we must have 
${m\choose \ell}/r  = \Theta(n/m)$. 
As $m=\omega(\kappa)$, it follows that ${m \choose \ell} = (\Theta(1)m/\ell)^\ell$.
Thus, $(\Theta(1)m/\ell)^\ell = \Theta(rn/m)$. 
Thus, we have that
%\[ v\le{m-\ell - \ell' \choose  \kappa - \ell - \ell'} / |\P| \quad \mbox{and}\quad |\P| = \Theta\left( {m \choose \kappa}\right) \quad \mbox{and} \quad t\ge 2^{k/2}.  \]
\[ v^{-1} \ge \frac{|P|}{{m-\ell - \ell' \choose  \kappa - \ell - \ell'}} \ge 
\frac{\Theta\left( {m \choose \kappa}\right)}{{m-\ell - \ell' \choose  \kappa - \ell - \ell'}} = 
\Theta(1) \cdot \frac{\frac{m!}{\kappa!(m-\kappa)!}}{\frac{(m-\ell-\ell')!}{(\kappa-\ell-\ell')!(m-\kappa)!}}\ge 
\Theta(1)\cdot \frac{m^{\ell+\ell'}}{(2\kappa)^{\ell+\ell'}} = 
\Theta(1)\cdot \frac{n}{m \Theta(1)^{\kappa}} \cdot \left( \frac{m}{2\kappa} \right)^{\ell'}\]
where the last step follows from  $(\Theta(1)m/\ell)^\ell = \Theta(rn/m)$, $r=2^{\kappa/3}$ and $\ell = \kappa/2$.

Now we bound $m$ in terms of $n$.
From $(m/(2\kappa))^\ell = \frac{n}{m c^\kappa}$ we obtain that
$m = \kappa^{\ell/(\ell+1)} \cdot n^{1/(\ell+1)}/\Theta(1)^\kappa = n^{2/(\kappa+2)} / \Theta(1)^\kappa$.
Remember that $\ell' = \Theta(\kappa/\log \kappa)$.
Based on this, we get that
$v^{-1} = n \cdot n^{\Theta(1/\log \kappa)} / \Theta(1)^{\kappa}$ and since $\kappa = o(\log n / \log\log n)$
the $\Theta(1)^{\kappa}$ term is dominated and we have
$v^{-1} = n \cdot n^{\Theta(1/\log \kappa)}$.
It remains to handle the extra $2^{-O(\beta)}$ factor in the space lower bound.
From Lemma~\ref{lem:nice}, we know that $\beta = c \log_\kappa m$.
Based on the value of $m$, this means $\beta = \Theta(\log n / (\kappa\log \kappa))$ which means
$2^{-O(\beta)}$ is also absorbed in $n^{\Theta(1/\log \kappa)}$ factor.
It remains to verify one last thing: previously, we claimed that we would verrify that $t = \omega(m)$.
Using the bound $t = \omega(2^{\kappa/2})$ this can be written as 
$2^{\kappa/2} = \omega(m)$ which translates to
$\kappa/2 = (2\log n)/(\kappa+2)  + \omega(1)$ which clearly holds if $\kappa \ge 3\sqrt{\log n}$.

\subsection{The Space Lower Bound}\label{subsec:spacewci}
We defer the details of our space lower bound to  Appendix~\ref{sec:wcis}, where we prove the following.

\begin{restatable}{theorem}{thmwcis}\label{thm:wcis}
	Any pointer-machine data structure that answers WCI queries with $\kappa$ wild cards in
	time $Q(n) + O(m+t)$ over an input of size $n$ must use 
	$\Omega\left( \frac{n}{\kappa} \Theta\left(\frac{ \log_{Q(n)} n}{\kappa}\right)^{\kappa-1}\right)$ space, 
	as long as $\kappa < \log_{Q(n)} n$,
	where $t$ is the output size, and $m$ is the pattern length.
\end{restatable}
Refer to the introduction for a discussion of the consequences of these lower bounds.

%% file: two_pattern.tex
\section{Two Pattern Document Indexing and Related Problems}

Due to lack of space we only state the primary results for 2P, FP, 2FP, and SI.
The proofs for Theorems \ref{thm:reportmain} and \ref{thm:searchingmain} can be found in Appendix \ref{sec:reporting_lb} and \ref{appendix:semi_group_lower_bound} respectively.

\begin{restatable}{theorem}{thmreportmain}\label{thm:reportmain}
Any data structure on the Pointer Machine for the 2P, FP, 2FP, and SI 
problems with query time $Q(n)$ and space usage $S(n)$ must obey $S(n)Q(n) = \Omega\left(n^{2-o(1)}\right)$.
\\
Also, if query time is  $O( (nk)^{1/2-\alpha} + k)$ for a constant $0 < \alpha < 1/2$, then
$S(n) = \Omega\left(n^{\frac{1+6\alpha}{1+2\alpha}-o(1)}\right)$.
\end{restatable}

The above theorem is proved using Theorem~\ref{thm:lblemma} which necessitates a randomized construction
involving various high probability bounds. Unlike our lower bound for $\kappa$-GPI we were unable to find a 
deterministic construction that uses Theorem~\ref{thm:dual}.

We also prove the following lower bound in the semi-group model which addresses the difficulty of the 
counting variants of 2P and the related problems.
\begin{restatable}{theorem}{thmsearchmain}\label{thm:searchingmain}
      Answering 2P, FP, 2FP, and SI  queries in the semi-group model requires $S(n) Q^2(n) = \Omega(n^2/ \log^4 n)$.
\end{restatable}

%% file: conclusions.tex
\section{Conclusions}\label{sec:conc}
In this paper we proved unconditional and high space and query lower bounds for a number of problems in string indexing. 
Our main message is that the pointer machine model remains an extremely useful tool for
proving lower bounds, that are close to the \emph{true complexity} of many problems.
We have successfully demonstrated this fact in the area of string and document indexing. 
Within the landscape of lower bound techniques, the pointer machine model, fortunately or unfortunately, 
is the only model where we can achieve \emph{unconditional, versatile, and high} lower bounds and we believe 
more problems from the area of string and document indexing deserve to be considered in this model.
To this end, we outline a number of open problems connected to our results.
\begin{enumerate}
    \item Is it possible to generalize the lower bound for 2P to the case where the two patterns are
		required to match within distance $\gamma$? This is essentially a "dual" of the $1$-GPI problem.
    \item Recall that our space lower bound for the WCI problem (Subsection~\ref{subsec:spacewci}) assumes
        that the  query time is independent of the alphabet size $\sigma$. 
        What if the query is allowed to increase with $\sigma$?
    \item Our query lower bound for the WCI (Subsection~\ref{subsec:querywci}) is proved for a binary alphabet.
        Is it possible to prove lower bounds that take $\sigma$ into account?
        Intuitively the problem should become more difficult as $\sigma$
        increases, but we were unable to obtain such bounds.
    \item We require certain bounds on $\kappa$ for the WCI problem. Is it possible to remove
		or at least loosen them? Or perhaps, can the upper bounds be substantially improved? 
    \item What is the difficulty of the $\kappa$-GPI problem when $\kappa$ is large?
\end{enumerate}

%% file: appendix_main.tex
\appendix
\input{appchazelle}

\lemnice % one section for proof of lemnice

\input{appendix_spacelb_wildcard}

\input{reporting_lb}

\input{appendix}

\input{app_semigroup_model}

\input{semi_group_appendix.tex}

%% file: appchazelle.tex
\section{Proof of Lemma~\ref{thm:lblemma}}
\label{sec:chazelle}

\thmlb*

  The proof is very similar to the one found in \cite{ChazelleR95},
  but we count slightly differently to get a better dependency on
  $\beta$. Recall the data structure is a graph where each node stores
  $2$ pointers and some input element. 
  At the query time, the algorithm must explore a subset of the graph.
  The main idea is to show  that the subsets explored by different queries
  cannot overlap too much, which would imply that there must be many
  vertices in the graph, i.e., a space lower bound.

  By the assumptions above a large fraction of the
  visited nodes during the query time will be output nodes (i.e., the algorithm
  must output the value stored in that memory cell). 
  We count the number of nodes in the graph by partitioning each query into sets
  with $\beta$ output nodes. By assumption 2 each such set will at
  most be counted $\ell$ times. We need the following fact:
  \begin{fact}[\cite{Afshani12}, Lemma 2]
    Any binary tree of size $ct$ with $t$ marked nodes can be
    partitioned into subtrees such that there are $\Theta(\frac{ct}{\beta})$ subtrees each with
    $\Theta(\beta)$ marked nodes and size $\Theta(\frac{ct}{\beta})$,
    for any $\beta \ge 1$.
  \end{fact}
  In this way we decompose all queries into these sets and count
  them. There are $|\Q|$ different queries, each query gives us
  $\frac{ct}{\beta}$ sets. Now we have counted each set at most
  $\ell$ times, thus  there are at least $\frac{ct|Q|}{\beta
    \ell}$ distinct sets with $\beta$ output nodes. On the other
  hand we know that starting from one node and following at most
  $a\beta$ pointers we can reach at most $2^{O(a\beta)}$ different
  sets (Catalan number). In each of those sets there are at most
  $\binom{a\beta}{\beta}$ possibilities for having a subset with
  $\beta$ marked nodes. This gives us an upper bound of
  $S(n)2^{O(a\beta)}\binom{a\beta}{\beta}$ for the number of possible
  sets with $\beta$ marked nodes. In conclusion we get
  \begin{equation*}
    S(n)2^{O(a\beta)}\binom{a\beta}{\beta} \ge \frac{ct|\Q|}{\beta\ell} \Rightarrow S(n) = \Omega\left(\frac{t|\Q|}{\beta\ell2^{O(a\beta)}}\right) = \Omega\left(\frac{t|\Q|}{\ell2^{O(a\beta)}}\right)
  \end{equation*}

%% file: appendix_spacelb_wildcard.tex
\section{Space Lower bound for $k$ wild card}\label{sec:wcis}

For this problem we use Chazelle's framework (Lemma~\ref{thm:lblemma})
and give a randomized construction for the input and query sets.
Note that here we no longer assume a binary alphabet.
In fact, we will vary the alphabet size.
To be more precise, we assume given any input of size $n$, over an alphabet of 
size $\sigma$, there exists a data structure that can answer WCI queries of size $m$ 
with $\kappa$ wild cards in $Q(n,\kappa) + O(m+t)$ time where $t$ is the output size.
Note that this query time is forced to be independent of $\sigma$.

Unlike the case for the query lower bound, we build the set of \textit{queries} in two stages.
In the first stage, we consider all documents of length $m$ over the alphabet $[\sigma]$ (that is $[\sigma]^m$) and
independently sample $n/m$ of them (with replacement) to form the initial set $\mathcal{D}$ of input documents.
And for queries, we consider the set $\mathcal{Q}$ of all strings of length $m$ over the alphabet $[\sigma]\cup \left\{ * \right\}$
containing exactly $\kappa$ wild cards (recall that $*$ is the wild card character).
In total we have $|\mathcal{Q}| = \binom{m}{\kappa}\tfrac{\sigma^m}{\sigma^\kappa}$ queries.
In the second stage, for a parameter $\beta$, we consider all pairs of queries and remove both queries if the 
number of documents they both match is $\beta$ or more.
No document is removed in this stage.
We now want to find a value of $\beta$ such that we retain almost all of our queries after the second stage.

The probability that a fixed query matches a random document is $\tfrac{\sigma^\kappa}{\sigma^m}$.
There are in total $|\mathcal{D}|$ documents, meaning we expect a query output $t = \tfrac{\sigma^\kappa}{\sigma^m}|\mathcal{D}|$ 
documents.
By an easy application of Chernoff bound we can prove that with high probability, all queries
output $\Theta(\tfrac{\sigma^\kappa}{\sigma^m}|\mathcal{D}|)$ documents.
%%To achieve this we apply the Chernoff bound. First fix a query $q \in \mathcal{Q}$.
%Now define random variables $x_1,x_2,\ldots,x_{|\mathcal{D}|}$ such that $x_i = 1$ if document $i$ does not match $q$.
%It follows that the expectation of $x_i$ is $1-\tfrac{\sigma^k}{\sigma^\tau}$, and thus the exptation of $X = \sum_i x_i$ is a $|\mathcal{D}|$ factor higher.

We now bound the number of  queries that survive the second stage.
First observe that if two queries do not have any wild card in the same position, 
then there is at most $1 \le \beta$ document that matches both.
Secondly, observe that for a fixed query $q$ there are $\binom{\kappa}{s}\sigma^s\binom{m-\kappa}{s}$ other queries 
sharing $\kappa-s$ wild cards.
We say these other queries are at distance $s$ from $q$.
For a fixed query, we prove that with constant probability it survives the second stage.
This is accomplished by considering each $s = 1,2,\ldots,\kappa$ individually and using a high concentration bound 
on each, and then using a union bound.
Since there are $\kappa$ different values for $s$ we bound the probability for each individual value by $\Theta(1/\kappa)$.

Now consider a pair of queries at distance $s$. The expected number of documents in their intersection is $\frac{t}{\sigma^s}$.
Letting $X$ to be the random variable indicating the number of documents in their intersection we get
\begin{equation*}
		\mathrm{Pr}[X > (1+\delta)\mu] = \mathrm{Pr}[X > \beta] = \left(\frac{e^\delta}{(1+\delta)^{1+\delta}}\right)^\mu < \left(\frac{t}{\sigma^s}\right)^{\Theta(\beta)}.
\end{equation*}
Recall that there are $\kappa$ values for $s$ and there are $\binom{\kappa}{s}\sigma^s\binom{m-\kappa}{s}$ ``neighbours'' at distance $s$, we want the following condition to hold:
\begin{equation*}
		\left(\frac{t}{\sigma^s}\right)^{\Theta(\beta)} \cdot \binom{\kappa}{s}\sigma^s\binom{m-\kappa}{s} \cdot \kappa \le \frac{1}{100} \Leftrightarrow \sigma^{{\Theta(s\beta)}} \ge 100 t^{\Theta(\beta)} \binom{\kappa}{s}\sigma^s\binom{m-\kappa}{s} \cdot \kappa
\end{equation*}
We immediately observe that the query time, $t$, should always be greater than $m$ (just for reading the query) and that there are never more than $\kappa \le m$ wild cards in a query.
Picking $\sigma = t^{1+\varepsilon}$ for some $\varepsilon > 0$ and letting $\beta$ be sufficiently large, 
we can disregard the factors $t^\beta$ and $\sigma^s$.
If $\sigma^{\Theta(s\beta)} > 100\kappa (\frac{e\kappa}{s})^s(\frac{m}{s})^s$ it follows that the condition above is satisfied.
Since $\kappa \le m \le t \le \sigma^{\frac{1}{1+\varepsilon}}$ it is sufficient to set $\beta = \Theta(1)$. 
We still have to derive the value of $m$.
Since $t = D \frac{\sigma^\kappa}{\sigma^m} = D \frac{t^{(1+\varepsilon)\kappa}}{t^{(1+\varepsilon)m}}$, $t \cdot t^{(1+\varepsilon)(m-\kappa)} = D$.
Manipulating this equation we see that $m = \kappa+\frac{\log_t D - 1}{1+\varepsilon}$.
We can now use Chazelle's framework (Lemma~\ref{thm:lblemma}):
by construction, the output size of any query is $t = \Omega(Q(n))$ and the any two queries have
$O(1)$ documents in common.
By the framework, we get the space lower bound of
\begin{equation*}
S(n) = \Omega\left(\binom{m}{\kappa}\frac{\sigma^m}{\sigma^\kappa}D\frac{\sigma^\kappa}{\sigma^m}\right) = \Omega\left(\frac{n}{m}\binom{m}{\kappa}\right).
\end{equation*}
For $\kappa < \log_{Q(n)} n$, we can upper bound $m \choose \kappa$ by $\Theta(\log_{Q(n)} n)^{\kappa}$.
Thus, we obtain the following theorem.
\thmwcis*

%% file: reporting_lb.tex
\section{Two Pattern Query Problem Reporting Lower Bound}
\label{sec:reporting_lb}

To prove the lower bound we use the (improved) framework of Chazelle
presented in Theorem~~\ref{thm:lblemma}.
To apply the lemma we need to create a set of queries and a set of inputs satisfying the stated properties. 
For now, we only focus on the 2P problem. 
The rest of this section is divided into four parts. 
The first part is preliminaries and technicalities we use in the remaining parts.
Next we describe how to create the documents, then we define the queries and finally we refine the two sets and prove that they satisfy the conditions stated in Theorem~\ref{thm:lblemma}.

\para{Preliminaries.} \label{sec:preliminaries}
Consider the alphabet $\Sigma = \left\{ \0, \1, \2, \cdots, \mathbf{2^\sigma-1} \right\}$ with $2^\sigma$ characters
(we adopt the convention to represent the characters in bold).
In our proof, documents and patterns are bitstrings.
For convenience we use $\bi$ to interchangeably refer to the character $\bi$ and its binary encoding.
The input is a set $\mathcal{D} = \{d_1, d_2, \ldots, d_D\}$ of $D$ documents where each document is a string over $\Sigma$ and the total length of the documents is $n$.
The set $\mathcal{D}$ is to be preprocessed such that given two patterns $P_1$ and $P_2$, that are also strings over $\Sigma$, all documents where both $P_1$ and $P_2$ occur can be reported.
Our main theorem is  the following.

\thmreportmain*
%% \begin{theorem}
%% Any data structure on the Pointer Machine for the Two Pattern Query
%% Problem with query time $Q(n)$ and space usage $S(n)$ must obey $S(n)Q(n) = \Omega\left(n^{2-o(1)}\right)$.

%% Also, if query time is  $O( (nk)^{1/2-\alpha} + k)$ for a constant $0 < \alpha < 1/2$, then
%% $S(n) = \Omega\left(n^{\frac{1+6\alpha}{1+2\alpha}-o(1)}\right)$.
%% \end{theorem}

We start by focusing on the 2P problem and in Appendix \ref{appendix:forbidden_modifications} and \ref{appendix:extensions} we describe the differences to obtain the same lower bounds for FP, 2FP and SI.

\ignore{
For technical reasons we encode bitstrings such that we can make sure
patterns align in a predictable manner in the documents. We have an
alphabet $\Sigma$ with $2^\sigma$ characters, and we encode  each character
using an \emph{encoding function}. The encoding is simple: given
$\sigma$ bits $b_1b_2\cdots b_\sigma$ it produces the bitstring
$110b_10b_2\cdots 0b_\sigma$, i.e., there is a prefix $11$ and every
bit $b_i$ is then prefixed by a $0$. This encoding also generalizes to
bitstrings of length $\ell > \sigma$ bits by applying the encoding to
chunks of size $\sigma$ and with the last chunk possibly having less than $\sigma$ bits:
given $\ell = k\sigma + r$ bits $b_1, \dots, b_\ell$, we get the bit string
$110b_10b_2\ldots0b_{\sigma}110b_{\sigma+1}\dots0b_{2\sigma}\ldots110b_{k\sigma+1}\dots0b_{k\sigma+r}$.
The encoding function is denoted by $\emph{enc}(\cdot)$. 
%Note that generating documents and queries by applying the encoding on bitstrings the patterns must align at the delimiting $11$s in the documents.
}

\para{The Documents.}
Let $\sigma$ be some parameter to be chosen later. 
In our construction, the first character, $\0$, works as a delimiter
and to avoid confusion we use the symbol $\bs$ to denote it. 
The number of documents created is $D$ which is also a parameter to be chosen later. The set of
documents is created randomly as follows. Each document will have $3(|\Sigma|-1)$ characters,
in $|\Sigma|-1$ consecutive parts of three characters each.
The $i$-th part, $1 \le i \le 2^\sigma-1$, is $\bs \bi b_1b_2\cdots b_\sigma$ where each $b_j, 1 \le j \le \sigma$ 
is uniformly and independently set to ``0'' or ``1''. 
In other words, the $i$-th part is the encoding of the delimiter (which is basically $\sigma$ ``0''s) 
followed by the encoding of $\bi$, followed by 
$\sigma$ random ``0''s and ``1''s.

\para{The Queries.}
The patterns in our queries always starts with $\bs$, 
followed by another character~$\bi$, $1 \le \bi < 2^\sigma$ (called the \emph{initial} character), followed by some trailing bits.
Observe that any pattern $P_1 = \bs\bi$, for $1 \le \bi < 2^\sigma$, matches all the documents. 
Observe also, if two patterns 
$P_1 = \bs\bi b_1\cdots b_p$ and $P_2 = \bs\bi b'_1\cdots b'_p$ where  $b_j, b'_j \in \{ 0,1 \}, 1 \le j \le p$,
match the same document, then we must have $b_j = b_j'$ for every $1 \le j\le p$.
Based on this observation, our set of queries, $\Q$, is all pairs of
patterns $(P_1, P_2)$ with different initial characters and $p$ trailing bits each, for some parameter $p$ to be set later. 
The previous observation is stated below.
\begin{ob}\label{ob:match}
    Consider two patterns $P$ and $P'$ that have the same initial characters.
    If a document matches both $P$ and $P'$, then $P=P'$ (i.e., they have
    the same trailing bits).
\end{ob}

A consequence of our construction is that, each pattern has one position where it can possibly match a document and matching only depends on the random bits after the initial character. 

\para{Analysis.}
We start by counting the number of queries. For each pattern in a query we have one initial character and $p$
trailing bits but the two initial characters should be different.  So the number of queries is
%\begin{equation*}
  $|\Q| = \binom{2^\sigma-1}{2}2^{2p} = \Theta(2^{2\sigma + 2p})$.
%\end{equation*}

%% We created the documents randomly and the queries deterministically.
%% We need to obey the properties stated in \ref{thm:lblemma} in the worst case, i.e. the intersection of some number of queries should be small. 
%% To accomplish this, we use high concentration bounds and claim with the right choice of parameters, the requirements of the lemma are satisfied with high probability.
%% For this we need to know the expected number of documents a query matches. 
%% Since the documents are independently built, the expected number of documents matching a particular query is just a factor of $D$ greater than the probability that a fixed query matches one randomly created document.

We need to obey the properties stated in Theorem~\ref{thm:lblemma} in the worst case,
i.e. the intersection of $\ell$ queries should have size at most  $\beta$.  The analysis
proceeds by using high concentration bounds and picking parameters carefully
such that the requirements of Theorem~\ref{thm:lblemma} are satisfied with high
probability.
First we study the probability of fixed patterns or fixed queries matching random documents.

\begin{lemma}\label{lemma:pr}
  The probability that a fixed pattern matches a random document is
  $2^{-p}$ and the probability that a fixed query matches a random document
  is $2^{-2p}$.
\end{lemma}
\begin{proof}
  The first part follows by the construction.
  Two patterns with distinct initial characters matching a document are independent events, thus the second part also follows.
\end{proof}

\begin{corollary}\label{cor:match}
  The expected number of documents matching a fixed query is $\Theta\left(D 2^{-2p}\right)$
\end{corollary}

According the framework, every query should have a large output 
(first requirement in Theorem~\ref{thm:lblemma}). 
By a straightforward application of Chernoff bounds it follows that all queries
output at least a constant fraction of the expected value. 
%Intuitively, this means, we want to pick $D$ ``large''  and $p$ ``small''. 
%However if $p$ is too small then it would limit the number of queries, which would give us a weaker lower bound. 
We also need to satisfy the second requirement in Theorem~\ref{thm:lblemma}. 
%We do this below.
Consider the intersection of $\ell^2$ queries, $q_1, \cdots, q_{\ell^2}$, for some parameter $\ell$.
Assume the intersection of $q_1,\cdots,q_{\ell^2}$ is not empty, since otherwise
the second requirement is already satisfied otherwise.
There must be at least $\ell$ distinct patterns among these queries thus it follows by Observation~\ref{ob:match} that there are at least $\ell$ distinct initial characters among the $\ell^2$ queries. 

\begin{ob}\label{ob:ell}
  The number of sets that contain $\ell$ patterns, $P_1, \cdots, P_\ell$, with distinct initial characters 
  is at most $(2^{\sigma+p})^\ell$.
\end{ob}

%We also need an extension of \rlem{pr}:
\begin{lemma}\label{lemma:matchell}
The probability that a random document satisfies a set of $\ell$ patterns, $P_1, \cdots, P_\ell$ 
is at most $\frac{1}{2^{p\ell}}$, and the expected number of such  documents is at most $\frac{D}{2^{p\ell}}$.
\end{lemma}
\begin{proof}
  This is just an extension of the case where $\ell=2$. If the patterns do not have distinct initial
  characters then the probability is 0 by Observation~\ref{ob:match}. Otherwise, the events for
  each constraint matching a document are mutually independent.
%% Let a set of $k$ constraints be given and choose a random
%% document. Each constraint must start at a unique position there are
%% $\binom{2^\sigma}{k}$ such positions. Each constraint is one character
%% followed by $p$ bits, which means for the following character there
%% are no more $2^{\sigma - p}$ choices. Note that some of these choices
%% are not possible due to picking starting positions for other
%% contraints, but it only weakens the bound to accept $2^{\sigma - p}$
%% of them. Now $2k$ characters have been chosen, and then permuting the
%% remaining characters gives us the bound:
%% \begin{eqnarray*}
%%   \Pr_{d}[d \textrm{ satisfies the } k \textrm{ constraints}] &\le& \frac{\binom{2^\sigma}{k}(2^{\sigma - p})^k(2^\sigma-2k)!}{(2^\sigma)!} \le \frac{2^{\sigma k}(2^\sigma - 2k)!}{2^{pk}(2^\sigma -k)!k!}\\
%% &\le& 
%% \frac{2^{\sigma k}(\frac{2^\sigma -2k}{e})^{2^\sigma -2k}\sqrt{2\pi(2^\sigma - 2k)}}{2^{pk}k!(\frac{2^\sigma - k}{e})^{2^\sigma - k}\sqrt{2\pi(2^\sigma -k)}}\\
%% &\le&
%% \frac{2^{\sigma k} e^k}{2^{pk}k!}
%% \left(\frac{(2^\sigma-2k)}{(2^\sigma-k)}\right)^{2^\sigma - 2k}
%% \frac{1}{(2^\sigma-k)^k}
%% \sqrt{\frac{2^\sigma-2k}{2^\sigma - k}} \\
%% &\le&
%% \frac{2^{\sigma k}e^k}{2^{pk}k!}
%% \left(1-\frac{k}{2^\sigma-k}\right)^{2^\sigma-2k}
%% \frac{1}{(2^\sigma-k)^k}
%% \sqrt{\frac{2^\sigma-2k}{2^\sigma - k}}
%% \end{eqnarray*}
\end{proof}

%% Assuming that $k < 2^\sigma / 2$ we continue the calculations:

%% \begin{eqnarray*}
%%   \Pr_{d}[d \textrm{ satisfies the } k \textrm{ constraints}] 
%% &\le& O\left(
%% \frac{2^{\sigma k}e^k}{2^{pk}k!}
%% \left(\frac{1}{e^k}\right)^{(2^\sigma -2k)/(2^\sigma - k)}
%% \frac{1}{(2^\sigma-k)^k}
%% \right) \\
%% &\le& O\left(
%% \frac{2^{\sigma k}}{2^{pk}}
%% \frac{1}{(2^\sigma-k)^k}
%% \right)
%% \\
%% &\le& O\left(
%% \frac{1}{2^{(p-1)k}}
%% \right)
%% \end{eqnarray*}

%% \begin{corollary}
%%   \label{cor:docs_sat_constraints}
%%   The expected number of documents satisfying at least $\ell$ constraints is $O(\frac{D}{2^{p\ell}})$.
%% \end{corollary}
%% \begin{proof}
%%   There are $D$ documents each having probability at most $\frac{1}{2^{p\ell}}$
%%   of satisfying at least $\ell$ constraints.
%% \end{proof}
We will choose parameters such that $\frac{D}{2^{p\ell}} = O(1)$. 
We need to consider the intersection of $\ell^2$ queries and bound the size of their intersection.

\begin{lemma}
\label{lemma:betasatell}
  The probability that at least $\beta$ documents satisfy $\ell$ given patterns is $O\left(\left(\frac{eD}{\beta2^{p\ell}}\right)^\beta\right)$.
\end{lemma}
\begin{proof}
  Let $X_i$ be a random variable that is 1 if the $i$-th document matches the given $\ell$ patterns and 0 otherwise, and 
  let $X = \sum_{i=1}^D X_i$. By \rlem{matchell} we have $\mathbb{E}[X] = \frac{D}{2^{p\ell}}$. By the Chernoff bound
%  \begin{equation*}
    $\Pr[X \ge \beta] =
    \Pr[X \ge (1+\delta)\mu] \le \left(\frac{e^\delta}{(1+\delta)^{1+\delta}}\right)^\mu \le \frac{e^\beta}{(\beta/\mu)^\beta}$
%  \end{equation*}
  where $\mu = D/2^{p\ell}$ and $1+\delta = \beta/\mu$.  The lemma follows easily afterwards.% we get $\left(\frac{eD}{\beta2^{p\ell}}\right)^\beta$.
\end{proof}

We now want to apply the union bound and use \rlem{betasatell} and Observation~\ref{ob:ell}.
To satisfy the second requirement of Theorem~\ref{thm:lblemma}, it suffices to have

\begin{equation}
    2^{\ell(p+\sigma)}O\left(\left(\frac{eD}{\beta2^{p\ell}}\right)^\beta \right) < 1/3. \label{eq:int}
\end{equation}

By Theorem~\ref{thm:lblemma} we have $S(n) \ge \frac{|\Q|t}{\ell^22^{O(\beta)}}$. 
Remember that by Corollary~\ref{cor:match}, the output size of each query is $t=\Omega( D2^{-2p})$. We set
$D$ such that $t = \Theta(Q(n))$.
We now plug in $t$ and $|\Q|$.
\begin{equation*}
  S(n) \ge \frac{|\Q|t}{\ell^22^{O(\beta)}} = \Omega\left(\frac{2^{2\sigma + 2p}D2^{-2p}}{\ell^22^{O(\beta)}}\right) = \Omega\left(\frac{(n/D)^2D}{\ell 2^{O(\beta)}} \right)= \Omega\left(\frac{n^2}{Q(n)2^{2p}\ell^22^{O(\beta)}}\right)
\end{equation*}

Since $D 2^\sigma = n$ (the input is $n$) and $D = \Theta(Q(n) 2^{2p})$. 
This implies $S(n)Q(n) = \Omega\left(\frac{n^2}{\ell^22^{2p}2^{O(\beta)}}\right)$, subject to 
satisfying~(\ref{eq:int}). 
Setting $p = \Theta(\beta) = \Theta(\sqrt{\log (n/Q(n))})$ and $\ell = \Theta(\log n)$ we satisfy the condition and get the trade-off:
%
%\begin{equation*}
  $S(n)Q(n) = \Omega\left(\frac{n^2}{2^{O(\sqrt{\log(n/Q(n))})}\log^2 n}\right)$
%\end{equation*}

Though we do not explicitly use $\sigma$ in the bound, one can verify that $\sigma = \Theta(\log(n/Q(n)))$ and thus we can assume that
each character fits in one memory cell.

Recently in the literature there has been bounds on the form
$O(\sqrt{nk}\log^{O(1)} n + k)$ with linear space. 
The trade-off proved here can be used to prove that is optimal
within polylogarithmic factors. 
Suppose $Q(n) = O((nk)^{1/2 - \alpha}+k)$, for a constant $0 < \alpha < 1/2$.
We can obtain a lower bound, 
by making sure that the search cost $(nk)^{1/2-\alpha}$ is dominated by the output size which means $k \ge
(nk)^{1/2-\alpha}$ or $ k \ge n^{(1/2-\alpha)/(1/2+\alpha)} =
n^{(1-2\alpha) / (1+2\alpha)}$. Plugging this in for the query time
gives the trade-off
$S(n) \ge \Omega({n^{\frac{1+6\alpha}{1+2\alpha}}}/ {2^{\Theta(\sqrt{\log(n)})}})$.

%% file: appendix.tex
\section{Forbidden Pattern lower bound modifications}
\label{appendix:forbidden_modifications}
One quickly sees that the inputs and queries designed for the 2P problem
do not work to prove lower bounds for the Forbidden Pattern case 
(i.e. one positive pattern and one negative pattern). 
To overcome this, we design a slightly different input
instance and the bounds are slightly different though still not more
than polylogarithmic factors off from the upper bounds.

The documents are created very similar to before, except each document
now comes  in two parts. The first part is the same as before,
i.e. $2^\sigma$ subsequences of $3$ characters where the $i$-th subsequence
is $\bs$, follow by $\bi$, followed by a trailing bits.
Now we extend the alphabet to
$\Sigma = \Sigma_1 \cup \Sigma_2 = [\mathbf{2^{\sigma+1}}]$, and the second
part of the document will only contain symbols from $\Sigma_2 =
\{\mathbf{2^\sigma}, \mathbf{2^\sigma + 1}, \ldots, \mathbf{2^{2\sigma} -1\}}$. 
The second part of a document is a uniformly random subset with $m$ characters from
$\Sigma_2$, however, as before, we prefix each character with $\bs$. 
The ordering of the characters in the second part does not
matter. The queries consist of a positive pattern and a negative
pattern. The positive patterns are created in the same way as before.
The negative pattern is just $\bs$ follow by one character from $\Sigma_2$.

\begin{fact*}
  There are $2^{2\sigma+p}$ different queries.
\end{fact*}
\begin{proof}
  A simple counting argument.
\end{proof}

\begin{fact*}
  The probability a fixed query hits a random document is $2^{-p}(1-\frac{m}{|\Sigma_2|})$ and
  the expected number of documents returned by a query is $\frac{D}{2^{p}}(1-\frac{m}{|\Sigma_2|})$
\end{fact*}
\begin{proof}
    The proof follows easily from the construction.
\end{proof}

For $\ell^2$ queries to have any output, there must either be
$\ell$ distinct positive patterns or $\ell$ distinct negative patterns.
In the former case, the rest of the proof from Appendix~\ref{sec:reporting_lb} goes through. In latter case there are
$\ell$ distinct negative patterns, i.e. $\ell$ distinct characters. If
any of the $\ell$ characters appear in a document, that document is
not in the intersection.

\begin{fact*}
  Let $P$ be a set of $\ell$ negative patterns. The probability that a
  random document contains at least one pattern from $P$ is
  $\frac{\binom{|\Sigma_2| - \ell}{m}}{\binom{|\Sigma_2|}{m}} =
  \Theta(\frac{1}{2^{pl}})$, for some choice of $m$.
\end{fact*}
\begin{proof}
Suppose we choose $1-\frac{m}{\Sigma_2}$ to be $2^{-p}$, then it
follows that $\Sigma_2 - m = \frac{\Sigma_2}{2^p}$. Using these equations it follows that
\begin{equation*}
  \frac{\binom{|\Sigma_2| - \ell}{m}}{\binom{|\Sigma_2|}{m}} = \frac{(\Sigma_2-m)(\Sigma_2-m-1)\cdots(\Sigma_2 - m - \ell + 1)}{\Sigma_2 (\Sigma_2 - 1)\cdots(\Sigma_2 - \ell + 1)} \le \frac{1}{2^{pl}}
\end{equation*}
\end{proof}

Now one can go through the remaining steps from \rlem{matchell} and onwards to get the same results.
However the values change slightly.
We see now $n = D(m+2^\sigma)$, which is still $\Theta(D2^\sigma)$, since $m$ is bounded by $2^\sigma$.
As noted above the number of queries has changed to $|Q| = 2^{2\sigma + p}$.
We get the following when plugging in the values in the framework.

\begin{equation*}
  S(n) \ge \frac{|Q|t}{\ell^22^{O(\beta)}} = \frac{2^{2\sigma + p}D2^{-2p}}{\ell^22^{O(\beta)}} = \frac{(n/D)^2D}{\ell^22^{O(\beta)}2^p} = \frac{n^2}{Q(n)\ell^2 2^{O(\beta)}2^{3p}}
\end{equation*}

By choosing the same values for $\beta$, $p$, and $\ell$ as before,
we get the same trade-off up to constant factors, the only difference
is the denominator has $2^{3p}$ factor instead of $2^{2p}$.

\para{Adapting to two negative patterns.}

The description above defines the hard instance for one positive
pattern and one negative pattern. The change was to split
each document in two parts, one regarding the positive element and one
regarding the negative element. In exactly the same way a hard
instance can be designed for two negative patterns: simply use the
negative construction for the both parts of the documents. There are
two minor details that also need to be addressed in the case of two
negative patterns. The first is, that the length of each document is
``long enough'', which easily follows since $m$ is at least a constant
fraction of $2^\sigma$. The second part, as also noted above, is that
the queries lose yet another $p$ bits, which means we can now only
create $2^{2\sigma}$ queries. Similarly to above, this means the
denominator will have a $2^{4p}$ factor rather than $2^{3p}$ (one
positive pattern, one negative pattern) or $2^{2p}$ (two positive
patterns).

\input{extensions}

%% file: extensions.tex
\section{Extensions}
\label{appendix:extensions}
\para{Set Intersection.}
We refer the reader to the introduction for the definition of this problem (and its variants).
Using a data structure that solves the set intersection problem,
one can solve the hard instance we gave for the 2P problem:
Notice that all our patterns have a fixed length.
This means in a suffix tree built on the set of documents, each pattern matches a set of leaves disjoint from all the other patterns.
For each pattern one identifies the set of documents in the corresponding leaf set as 
the input for the Set Intersection problem.
Verifying the correctness is rather straightforward.

\para{Common Colors.}  Here one is given an array of
length $n$ where each entry is a color (potentially with a weight) and the queries are two intervals where
the answer is the list of colors that appear in both intervals (reporting) or
the weighted sum of the colors (searching).
It is easy to see that a solution for the common colors problem can be used to solve 2P query problem~\cite{Ferragina2Dind} (using a suffix tree)
and thus our lower bounds apply.

\begin{corollary}
  The reporting (resp. searching) variant of the common colors problem or set intersection requires
  $S(n)Q(n) = \Omega(n^2/(2^{O(\sqrt{\log(n/Q(n)})}\log^2 n ))$ (resp. $S(n)Q^2(n) = \Omega(n^2/\log^4 n)$) in the pointer machine 
  (resp. semi-group) model.
\end{corollary}

\para{Two-Sided High-Dimensional Queries}
\noindent
\label{appendix:twosided}
Another implication is for high dimensional two-sided range queries, a special case of the fundamental problem of orthogonal range searching (see~\cite{aal12,aal10,Chazelle90a,Chazelle90b,c86,c88,clp11} for the best upper bounds, lower bounds, history and motivation).
In this problem, we would like to preprocess a set of points in $d$-dimensional space 
such that given an axis-aligned query rectangle $q=(a_1, b_1) \times \cdots \times (a_d, b_d)$, 
we can report (or count) the points contained in the query.
We consider a variant of the problem in which the input is in $d$-dimensions
but the query box $q$ has two sides, that is, all $a_i$'s and $b_i$'s are set to $-\infty$
and $+\infty$, respectively, except only for two possible values.

To give a lower bound for this problem, we consider our construction for the 2P problem,
and a suffix tree built on the set of documents. 
Consider a non-delimiter character $\bi$ and its corresponding range of leaves, $S_\bi$, in the suffix tree: 
we are guaranteed that each document appears exactly once in $S_\bi$.
We create $D$ points in $(2^\sigma-1)$-dimensional space by picking the $i$-th coordinate of the
points based on the order of the documents in $S_\bi$.
With a moment of thought, it can be verified that we can use a solution for 2-sided orthogonal
range searching to solve our hard instance for the 2P problem.
By reworking the parameters we obtain the following.
\begin{theorem}
  For a set of $m$ points in $d$-dimensional space, answering 2-sided orthogonal
  range reporting (resp. searching) queries in $m^{1-o(1)} + O(k)$, (resp. $m^{1-o(1)}$ ) time
  in the pointer machine (resp. semi-group) model requires $\Omega( n^{2-o(1)} / m)$
  (resp. $\Omega( n^{2-o(1)} / m^2)$) space, where $n=md$.
\end{theorem}

%% file: app_semigroup_model.tex
\section{The Semi-Group Model}
\label{appendix:semi_group}
In this model, each element of the input set $\U=\left\{ x_1, \dots, x_n \right\}$
is assigned
a weight, by a weight function $\w: \U \rightarrow G$ where $G$ is a semi-group.
Given a query $q$, let $q_\U$ be the subset of $\U$ that matches $q$. 
The output of the query is $\sum_{x \in q_\U} \w(x)$.
By restricting $G$ to be a semi-group, the model can capture very diverse set of queries. 
For example, finding the document with maximum weight that matches a query can
be modelled using the semi-group $\mathbb{R}$ with ``$\max$'' operation\footnote{
Observe that there is no inverse operation for $\max$.}.
The lower bound variant of the semi-group model was first
introduced by Chazelle~\cite{Chazelle90b} in  1990.
In this variant, it is assumed that the data structure stores ``precomputed'' sums $s_1, \dots, s_m$
where each $s_i = \sum_{j \in I_i} \alpha_{i,j}\w(x_j)$, $\alpha_{i,j} \in \mathbb{N}$, and $I_i$ is some subset of $[n]$ and thus
$m$ is a lower bound on space cost of the data structure.
While the integers $\alpha_{ij}$ can depend on the weight function $\w$, the set $I_i$ must be
independent of $\w$ (they both can depend on $\U$).
Finally, the semi-group is required to be \emph{faithful} which essentially means that if
$\sum_{i \in I} \alpha_i g_i = \sum_{i \in I'}\alpha'_i g_i$ for every assignment of semi-group
values to variables $g_i$, then we must have $I=I'$ and $\alpha_i = \alpha'_i$.
To answer a query, the data structure can pick $t$ precomputed values to create
the sum $\sum_{x \in q_\U} \w(x)$ and thus $t$ is a lower bound on the
query time.
For a detailed description and justification of the model we refer the reader to Chazelle's paper~\cite{Chazelle90b}.
Similar to the pointer machine model, this semi-group model is often used to prove lower bounds that
closely match the bounds of the existing data structures (see e.g.,~\cite{Chazelle90b,Chazellesimplb,hslb}).

%% file: semi_group_appendix.tex
\section{Semi-Group Model Lower Bound}
\label{appendix:semi_group_lower_bound}

%% We quickly recap the basics of the semi-group model again: we will assume  that
%% the data structure is allowed to store ``pre-computed sum'' and each sum
%% is associated with a subset of the input documents. The sum
%% is a linear combination of the weights of the associated  documents.
%% To answer a query, the data structure can recover some of these sums, and 
%% add them up to produce the final answer.
%% The number of sums stored is the space cost of the data structure and the
%% number of sums used at the query time is a lower bound on the query time. 
%% Note that the time to actually find these sums is not counted.

The general strategy for proving the lower bound is as follows. 
We create queries and documents where the answer for each query is sum of
``a lot of'' documents. 
Then we aim to prove that any subset of poly($\log n$) documents are
unlikely to have more than a few patterns in common, which subsequently implies,
any pre-computed sum that involves more than polylogarithmic weights 
is only useful for few queries; let's call such a pre-computed sum ``crowded''. 
We charge one unit of space for every ``crowded'' sum stored. 
At the query time however, since the answer to queries involves
adding up ``a lot of'' weights,  to reduce the query time, the query algorithm
must use at least one ``crowded'' sum.  But we argued that each ``crowded'' sum
can be useful only for a small number of queries which means to be able to
answer all the queries, the data structure should store many ``crowded''  
sums, giving us a lower bound.

This strategy is very similar to the strategy employed by Dietz el al. \cite{DietzMRU95} to prove a lower bound for the offline set intersection problem of processing a sequence of updates and queries.
However we need to cast the construction in terms of documents and two-pattern queries.
The results we get are for static data structure problems where as they care about an algorithm for processing a sequence of queries/updates.
This means the results are not easily comparable since there are other parameters, such as the number of queries and updates.
They do also provide a lower bound for dynamic online set intersections with restrictions on space. However, since it is dynamic the results are also not directly comparable.
Both their result and our result suffer the penalty of some $\textrm{poly}\log(n)$ factors, but ours are slightly more sensitive. Their log factors are only dependent on $n$ (the number of updates, which for us is the number of elements) where as some of our log factors depend on $Q(n)$ instead (the query time), which makes a difference for fast query times.
However when considering the pointer machine model we have to be much more careful than they are in considering the size of the intersection of many sets.
For the pointer machine it is absolutely crucial that the size is bounded by a constant, whereas they can get away with only bounding it by $O(\log n)$ and only suffer $\log n$ factor in the bounds.
If the intersection cannot be bounded better than $\varepsilon \log n$ for pointer machine lower bounds, the bound gets a penalty of $n^\varepsilon$, which is really detrimental and would make the lower bounds almost non-interesting. Luckily we are able to provide a better and tighter analysis in this regard.

\para{Documents and Queries.}
Our construction will be very similar to the one used in the reporting case and the
only difference is that the number of trailing bits, $p$ will be constant.
To avoid confusion and for the clarity of the exposition, we will still use the variable $p$.

\subsection{Lower bound proof}

We now wish to bound how many documents can be put in a ``crowded'' sum
and be useful to answer queries. 
To be specific, we define a ``crowded'' sum, as a linear combination of 
the weights of at least $\beta$ documents. We would like to show that 
such a combination is useful for a very small number of queries (i.e., patterns).
The same technique as the previous section works and we can in fact
continue from inequality~(\ref{eq:int}).
\ignore{
Not surprisingly, our proof techniques will be very similar. 
In fact, we can continue from In order for a subset of documents to
be useful for at least $\ell^2$ queries, they should all contain $\ell$ patterns $P_1,
P_2, \ldots, P_\ell$ (each starting with a different character), in this
case the documents are useful for $\Theta(t^2)$ different queries. The
probability that a document contains all these patterns is
$\left(\frac{1}{2^p}\right)^\ell$. If we create $D$ documents and define
\begin{equation*}
  x_i = \left\{
  \begin{array}{rl}
    1 & \textrm{if $P_1, P_2, \ldots, P_t$ all occur in $d_i$} \\
    0 & \textrm{otherwise}
  \end{array}
  \right.
\end{equation*}

Defining $X = \sum_{i=1}^D x_i$ we have the number of documents where
these $t$ patterns occur. In expectation we see for a fixed set of $t$
patterns that $\mathbb{E}[X] = \frac{D}{2^{pt}}$. Since each document
is chosen independently of all other documents we can apply a high
concentration bound to show that with high probability very few
documents contain these $t$ patterns. From the Chernoff bound we get:

\begin{equation}
  Pr[X > r = \mu \delta] \le \left(\frac{e^\delta}{\delta^\delta}\right)^\mu = \left(\frac{e\mu}{r}\right)^r
\end{equation}
}

We make sure that $2^{p\ell}\ge D^2$ or $ 2\log D \le p\ell$. 
We also choose $\beta > ce$ for some $c>1$ and simplify inequality~\ref{eq:int} to get,
$2^{\ell(p + \sigma )}  \le 2^{p\ell\beta/2} \Rightarrow{p + \sigma } \le {p\beta/2}$, or
\begin{align}
   \sigma  &\le {p\beta/2}.  \label{eq:count}
\end{align}

If the above inequality holds, then, no sets of  $\beta$
documents can match $\ell$ different patterns. 
Furthermore, remember that with high probability the number of documents that match every query is $\frac{D}{2\cdot 2^{2p}}$.
Observe that to answer a query $q$ faster than $\frac{D}{2\cdot 2^{2p}\beta}$ time, one has to use at least one
crowded sum $s$.
If there is a document that is used in $s$ but it does not match $q$,
then the result will be incorrect as in the semi-group model there is no way to 
``subtract'' the extra non-matching weight from $s$. 
Thus, all the documents whose weights are used in $s$ must match $q$.
However, we just proved that documents in a crowded sum can only match $\ell$ patterns,
and thus be useful for $\ell^2$ queries.
So, if inequality~(\ref{eq:count}) holds, either
$Q(n) \ge \frac{D}{2 \cdot 2^{2p}\beta}$ or $S(n) \ge
\frac{2^{2\sigma + 2p}}{\ell^2}$; the latter follows since $2^{2\sigma + 2p}$ is the total number of queries
and the denominator is the maximum number of queries where a crowded sum can be used for. 
To exclude the first possibility, we pick $D$ such that 
\begin{align}
    Q(n) <  \frac{D}{2 \cdot 2^{2p}\beta}\label{eq:count2}.
\end{align}
Recall that $n = D 2^\sigma$. We pick $\sigma = \log({n}/{D})$, $p=O(1)$, which forces 
$\beta  \ge \log({n}/{D)})$ and $\ell \ge \log D = \Theta(\log (n/2^\sigma))
= \Theta(\log Q(n))$. 
We also pick $D$ as small as possible while satisfying (\ref{eq:count2}).
This leads to the following trade-off:

\begin{equation*}
  S(n) \ge \frac{2^{2\sigma + 2p}}{\ell^2} = \frac{n^2/D^2}{\Theta(\log^2 Q(n))}= \frac{n^2}{\Theta(Q^2(n)\log^2n \log^2 Q(n))}\Rightarrow S(n)Q(n)^2 = \Omega(n^2/\log^4 n).
\end{equation*}

\thmsearchmain*